\documentclass[onecolumn]{IEEEtran}
\IEEEoverridecommandlockouts

\usepackage[utf8]{inputenc} 
\usepackage[T1]{fontenc}
\usepackage{url}
\usepackage{ifthen}
\usepackage{cite}
\usepackage[cmex10]{amsmath}
\usepackage{bm}

\usepackage{bbm}
\usepackage{amssymb,amsfonts}
\usepackage{mathtools}
\usepackage{amsthm}
\usepackage{algorithm}
\usepackage{algpseudocodex}
\usepackage{graphicx}
\usepackage{textcomp}
\usepackage{xcolor}
\usepackage{cases}
\usepackage{setspace}
\usepackage{enumerate}
\usepackage{enumitem}

\usepackage{hyperref}
\usepackage{stmaryrd}
\usepackage{balance}

\newtheorem{theorem}{Theorem}

\newtheorem{lemma}[theorem]{Lemma}
\newtheorem{claim}[theorem]{Claim}
\newtheorem{prop}[theorem]{Proposition}

\newtheorem{defn}{Definition}

\theoremstyle{remark}
\newtheorem{remark}{Remark}[section]

\DeclareMathOperator{\wt}{wt}
\newcommand{\llb}{\llbracket}
\newcommand{\rrb}{\rrbracket}
\DeclareMathOperator{\med}{med}
\newcommand{\cev}[1]{\overset{{}_{\shortleftarrow}}{#1}}

\begin{document}
	
	\title{
		Reconstruction of multiple strings of constant weight from prefix-suffix compositions\\
		\thanks{
			\hspace{-1.2em}\rule{1.5in}{.4pt}
			
			This paper was presented in part at the 2024 IEEE International Symposium on Information Theory (ISIT 2024) \cite{yang2024reconstruction}.
			
			Y. Yang was with the School of Data Science, The Chinese University of Hong Kong, Shenzhen, China. (Email: yaoyuyang@link.cuhk.edu.cn)			
			
			Z. Chen is with School of Science and Engineering, Future Networks of Intelligence Institute, The Chinese University of Hong Kong, Shenzhen, China. (Email: chenztan@cuhk.edu.cn) 
			This research was supported in part by the Basic Research Project of Hetao Shenzhen-Hong Kong Science and Technology Cooperation Zone under Project HZQB-KCZYZ-2021067, 
			the Guangdong Provincial Key Laboratory of Future Network of Intelligence under Project 2022B1212010001, 
			the National Natural Science Foundation of China under grants 62201487,
			and the Shenzhen Science and Technology Stable Support Program.
		}
	}
	
	\author{%
		\IEEEauthorblockN{Yaoyu Yang} \hspace*{1in}
		\and
		\IEEEauthorblockN{Zitan Chen}
	}
	
	\maketitle
	
	\begin{abstract}
		Motivated by studies of data retrieval in polymer-based storage systems, we consider the problem of reconstructing a multiset of binary strings that have the same length and the same weight from the compositions of their prefixes and suffixes of every possible length. We provide necessary and sufficient conditions for which unique reconstruction up to reversal of the strings is possible. Additionally, we present two algorithms for reconstructing strings from the compositions of prefixes and suffixes of constant-length constant-weight strings.
	\end{abstract}
	
	
	\section{Introduction}
	The growing demand for archival data storage calls for innovative solutions to store information beyond traditional methods that rely on magnetic tapes or hard disk drives. Recent advancement in macromolecule synthesis and sequencing suggests that polymers such as DNA are promising media for future archival data storage, largely attributed to their high storage density and durability. Data retrieval in polymer-based storage systems depends on macromolecule sequencing technologies \cite{al2017mass,launay2021precise} to read out the information stored in polymers. However, common sequencing technologies often only read random fragments of polymers. Thus, the task of data retrieval in these systems has to be based on the information provided by the fragments. 
	
	Under proper assumptions, one may represent polymers by binary strings and turn the problem of data retrieval into the problem of string reconstruction from substring compositions, i.e., from the number of zeros and the number of ones in substrings of every possible length. In \cite{acharya2015string}, the authors characterized the length for which strings can be uniquely reconstructed from their substring compositions up to reversal.
	Extending the work of \cite{acharya2015string}, the authors of \cite{pattabiraman2023coding} and \cite{banerjee2023insertion} studied the problem of string reconstruction from \emph{erroneous} substring compositions. Specifically, \cite{pattabiraman2023coding} designed coding schemes capable of reconstructing strings in the presence of substitution errors and \cite{banerjee2023insertion} further proposed codes that can deal with insertion and deletion errors.
	Observing that it may not be realistic to assume the compositions of all substrings are available, the authors of \cite{gabrys2022reconstruction} initiated the study of string reconstruction based on the compositions of prefixes and suffixes of all possible lengths. In fact, \cite{gabrys2022reconstruction} considered the more general problem of reconstructing \emph{multiple} distinct strings of the same length simultaneously from the compositions of their prefixes and suffixes. The main result of \cite{gabrys2022reconstruction} reveals that for reconstruction of no more than $h$ distinct strings of the same length, there exists a code with rate approaching $1/h$ asymptotically. Following \cite{gabrys2022reconstruction}, the authors of \cite{ye2023reconstruction} studied in depth the problem of reconstructing a single string from the compositions of its prefixes and suffixes. In particular, their work completely characterized the strings that can be reconstructed from its prefixes and suffixes compositions uniquely up to reversal.
	
	Efficiency of data retrieval is a major concern for practical polymer-based storage systems, and thus low complexity algorithms for string reconstruction are of great interest. In the case of reconstruction from error-free substring compositions, \cite{acharya2015string} described a backtracking algorithm for binary strings of length $n$ with worst-case time complexity exponential in $\sqrt{n}$. Moreover, \cite{pattabiraman2023coding} and \cite{gupta2022new} constructed sets of binary strings that can be uniquely reconstructed with time complexity polynomial in $n$.
	In the case of reconstruction from error-free compositions of prefixes and suffixes, \cite{gabrys2022reconstruction} and \cite{ye2023reconstruction} presented sets of binary strings that can be efficiently reconstructed. For reconstruction in the presence of substitution composition errors, \cite{pattabiraman2023coding} showed that when the number of errors is a constant independent of $n$, there exist coding schemes with decoding complexity polynomial in $n$. 
	
	We note that string reconstruction is classic problem \cite{margaritis1995reconstructing,levenshtein2001efficient,batu2004reconstructing} and has been studied under various settings, including reconstruction from substrings \cite{margaritis1995reconstructing,marcovich2021reconstruction,yehezkeally2023generalized} and from subsequences \cite{levenshtein2001efficient,batu2004reconstructing,cheraghchi2020coded,krishnamurthy2021trace,ravi2022coded,levick2023fundamental} under either combinatorial or probabilistic assumptions.
	
	In this paper, we consider the problem of reconstructing $h$ strings that are not necessarily distinct but have the same length $n\geq 1$ and weight $\bar{w}\leq n$ from their error-free compositions of prefixes and suffixes of all possible lengths. The problem of reconstructing multiple strings from prefix-suffix compositions becomes more amenable to analysis if the strings are of constant weight. This is because nice properties due to symmetry can be tethered to the prefix-suffix compositions. Our first result is a characterization of the properties of constant-weight strings that enable unique reconstruction up to reversal. Additionally, we present two algorithms that reconstruct constant-weight strings from prefix-suffix compositions. Given prefix-suffix compositions as input, one of the algorithms can efficiently output a multiset of strings whose prefix-suffix compositions are the same as the input, and the other is able to output all multisets of strings up to reversal that are allowed by the input. Our analysis relies on the running weight information of the strings that can be extracted from the prefix-suffix compositions and the inherent symmetry of constant-weight strings and their reversals.
	
	The rest of this paper is organized as follows. In Section~\ref{sec:n}, we present the problem statement and introduce necessary notation and preliminaries that are helpful for later sections. In particular, we introduce the notion of cumulative weight functions that capture the running weight information of a multiset of strings, which is used throughout the paper. In Section~\ref{sec:conditions}, we derive the necessary and sufficient conditions for unique reconstruction. Section~\ref{sec:algorithm} is devoted to the reconstruction algorithms. We conclude this paper and mention a few open problems in Section~\ref{sec:conclusion}.

	\section{Notation and preliminaries}\label{sec:n}
	Let $n$ be a positive integer. Denote $[n]=\{1,2,\ldots,n\}$ and $\llb n\rrb=\{0,1,\ldots,n\}$. For integers $n_1, n_2$, define $[n_1,n_2]=\{n_1,n_1+1,\ldots,n_2\}$ if $n_1\leq n_2$ and $[n_1,n_2]=\emptyset$ if $n_1>n_2$.
	Let $\bm{t}=t_1t_2\ldots t_n \in \{0,1\}^n$ be a binary string of length $n$ and the \emph{reversal} of $\bm{t}$ is denoted by $\cev{\bm{t}}=t_nt_{n-1}\ldots t_1$. The \emph{weight} of $\bm{t}$ is the number of ones in $\bm{t}$, denoted by $\wt(\bm{t})$. The \emph{composition} of $\bm{t}$ is formed by the number of zeros and the number of ones in $\bm{t}$. More precisely, the ordered pair $(n-\wt(\bm{t}),\wt(\bm{t}))$ is called the composition of $\bm{t}$. 
	For $1\leq l\leq n$, the length-$l$ prefix and the length-$l$ suffix of $\bm{t}$ are denoted by $\bm{t}[l]$ and $\bm{t}[-l]$, respectively. 
	\begin{defn}
		The {set} of compositions of all prefixes of a string $\bm{t}\in\{0,1\}^n$ is called the prefix {compositions} of $\bm{t}$, denoted by $M_p(\bm{t})$. More precisely, 
		\begin{align*}
			M_p(\bm{t})=\{ (j-\wt(\bm{t}[j]),\wt(\bm{t}[j])) \mid 1\leq j\leq n\}.
		\end{align*}
		The suffix {compositions} of $\bm{t}$ are {similarly defined to be}
		{\begin{align*}
				M_s(\bm{t})=\{ (j-\wt(\bm{t}[-j]),\wt(\bm{t}[-j])) \mid 1\leq j\leq n\}.
			\end{align*} 
			The} prefix-suffix compositions of $\bm{t}$ are defined to be the multiset union of $M_p(\bm{t})$ and $M_s(\bm{t})$, denoted by $M(\bm{t})$.\footnote{We will use ``$\cup$'' to denote both the set union and the multiset union. The exact meaning of ``$\cup$'' will be clear from the context.}
		Let {$U$ be a multiset of binary strings}. Define {$M_p(U)$ to be the multiset union of $M_p(\bm{t}),\bm{t}\in U$, i.e.,
			\[
			M_p(U)=\bigcup_{\bm{t}\in U}M_p(\bm{t}).
			\]
			The multiset $M_s(U)$ is defined similarly. The multiset $M(U):=M_p(U)\cup M_s(U)$ is called the prefix-suffix compositions of $U$.}
	\end{defn}
	
	Note that different multisets may result in the same prefix-suffix compositions. For example, reversing a string in a multiset gives rise to a different multiset that has the same prefix-suffix compositions. 
	\begin{defn}\label{def:reversal} 
		Let $U$ and $V$ be two multiset of strings. The multiset $V$ is said to be a \emph{reversal} of $U$, denoted by $V\sim U$, if {$|V|=|U|$, and for any string $\bm{t}\in U$ the sum of the multiplicities of $\bm{t}$ and $\cev{\bm{t}}$ in $U$ equals the sum of the multiplicities of $\bm{t}$ and $\cev{\bm{t}}$ in $V$.}
		The collection of multisets that are reversals of $U$ forms an equivalent class, denoted by $[U]$, i.e., $[U]:=\{V\mid V\sim U\}$.
	\end{defn}
	
	Given the prefix-suffix compositions of a multiset $H$ of $h\geq 1$ binary strings of length $n$, we are interested in finding the $h$ strings in $H$. In the sequel, we denote $M:=M(H)$ for simplicity. Clearly, any reversal of $H$ has prefix-suffix compositions $M$. However, there may exist multisets that are not reversals of $H$ but have the same compositions as $H$. If a multiset has prefix-suffix compositions $M$, we say the multiset is \emph{compatible} with $M$.
	Let $\mathcal{H}:=\{[U]\mid M(U)=M\}$ be the collection of all equivalent classes whose members are compatible with $M$. {We say $H$ can be uniquely reconstructed up to reversal if and only if $|\mathcal{H}|=1$.}

	As we will see later, it is helpful to present the information provided by each composition, i.e., the length and the weight of the corresponding substring, on a two dimensional grid. This motivates the following notation. 
	Note that since $|H|=h$ there are $2nh$ compositions in $M$. Denote the grid by \[{T:=\{(l,m)\mid l\in\llb n\rrb, m\in[2h]\}.}\] 
	{Assume the strings in $H$ are given by $\bm{h}_j$, $j=1,\ldots,h$. Then one can record $\wt(\bm{h}_j[l])$ on the grid $T$ with coordinates $(l,2j-1)$ and $\wt(\bm{h}_j[-l])$ on the grid $T$ with coordinates $(l,2j)$.} So the task of reconstructing $H$ from $M$ becomes appropriately identifying the second coordinate of $(l,m)\in T$, i.e., the label of the string in $H$, based on the weights of the prefixes and suffixes. To this end, we define an integer-valued bivariate function on $T$, stated below.
	
	\begin{defn}\label{def:1}
		A function $f \colon T \to \llb n\rrb$ is called a cumulative weight function (CWF) if it satisfies the following conditions:
		\begin{enumerate}[label=(\roman*)]
			\item $f(0,m) = 0$ for any $m \in [2h]$; \label{ite:def:1:1}
			\item $f(l,m)-f(l-1,m)\in\{0,1\}$ for any $ (l,m) \in [n]\times [2h]$; \label{ite:def:1:2}
			\item {for each $j \in [h]$, there exists $w_j \in \llb n \rrb$ such that $f(l,2j-1)+f(n-l,2j) = w_j$ for all $l \in \llb n \rrb$.} \label{ite:def:1:3}
		\end{enumerate}
		{If {$w_j=\bar{w}\in \llb n\rrb$ for all $j \in [h]$}, then $f$ is said to be a constant-weight CWF or have constant weight $\bar{w}$.}
	\end{defn}
	
	{It is clear} that a CWF can be induced by the weights of the prefixes and suffixes of the strings in $H$. In particular, Item~\ref{ite:def:1:3} in Definition~\ref{def:1} is satisfied by taking {$w_j = \wt(\bm{h}_j)$}. At the same time, a CWF also identifies a multiset of $h$ strings because one can reconstruct string {$\bm{h}_j$ using the weights of the prefixes given by $\{f(l,2j-1)\mid l \in \llb n \rrb\}$} straightforwardly.
	
	\begin{defn}\label{def:2}
		Let $f \colon T \to \llb n\rrb$ be a CWF. The multiset {$H_f := \{\bm{t}_j={t}_{j,1}\ldots{t}_{j,n} \mid t_{j,l} = f(l,2j-1)-f(l-1,2j-1) \text{ for all } l \in [n], j \in [h] \}$} is called the multiset of strings corresponding to CWF $f$.
	\end{defn}
	
	Note that a CWF $f$ uniquely determines $H_f$, and the multiset $H$ (or any of its reversals) induces CWFs that are equivalent up to permutation of {the ordering} of the strings in $H$. So one may use CWFs as a proxy for analyzing the reconstructibility of $H$ based on $M$.
	By definition, a CWF $f(l,m)$ consists of $2h$ univariate functions obtained by fixing the variable $m\in [2h]$. It is convenient to deal with these component functions directly. 
	
	\begin{defn}
		Let $f\colon T\to \llb n\rrb$ be a CWF. For $m\in[2h]$, let $f_m\colon \llb n\rrb \to \llb n\rrb$ be the function given by $f_m(l)=f(l,m)$.
	\end{defn}
	
	By Item~\ref{ite:def:1:3} of Definition~\ref{def:1}, if $f$ is the CWF induced by $H$, then $f_{2j-1}$ and $f_{2j}$ record the weight information starting from the two ends of the same string in $H$. In other words, $2j-1$ and $2j$ refer to the same string. As we will be constantly relating $f_{2j-1}$ to $f_{2j}$ or the other way around, let us introduce the following definition for notational convenience.
	\begin{defn}
		Let $m\in[2h]$. Define $m^*\in [2h]$ by
		\begin{align*}
			m^*=\begin{cases}
				m-1 \text{ if $m$ is even},\\
				m+1 \text{ if $m$ is odd}.
			\end{cases}
		\end{align*}
	\end{defn}
	
	The problem of reconstructing a single string from its prefix-suffix compositions, i.e., the case where $h=1$ in our setting, is examined in \cite{ye2023reconstruction}. The authors of \cite{ye2023reconstruction} introduced the so-called swap operation for a string $\bm{t}$ to generate all the strings that have the same prefix-suffix compositions as $\bm{t}$, thereby deducing the conditions for a single {string} to be reconstructed uniquely up to reversal from prefix-suffix compositions. Specifically, the swap operation is performed on carefully chosen coordinates where $\bm{t}$ and $\cev{\bm{t}}$ disagree, so as to produce new strings that maintain the same prefix-suffix compositions. 
	Let $f$ be the CWF induced by $\{\bm{t}\}$ and let $f_1,f_2$ correspond to $\bm{t}$ and $\cev{\bm{t}}$, respectively. Using the language of CWFs, the swap operation should be performed over the domain where $f_1$ and $f_2$ take different values. Since $f_1,f_2$ capture the running weight information from the two ends of the same string $\bm{t}$, they must be $180$-degree rotational symmetric. More precisely, if $\wt(\bm{t})=\bar{w}$, then $f_1$ should be the same as $f_2$ when it is rotated 180 degrees about $(n/2,\bar{w}/2)$. With this observation, it follows that if $f'_1$ and $f'_2$ are the functions corresponding to the strings obtained by swapping bits of $\bm{t}$ with $\cev{\bm{t}}$, then $f_1',f_2'$ must be $180$-degree rotational symmetric for them to record the weight information from the two ends of a single string.
	
	Generalizing the idea of comparing $\bm{t}$ and $\cev{\bm{t}}$ for producing new strings, we introduce the notions of discrepancy and maximal intervals between functions $f_{m_1}$ and $f_{m_2}$ for any $m_1,m_2 \in [2h]$ and $h \geq 1$ as follows.
	
	\begin{defn}\label{def:compare}
		For $m_1,m_2\in[2h]$, define the \emph{discrepancy} between the functions $f_{m_1}$ and $f_{m_2}$ to be {the set $D(m_1, m_2):=\{ l \in [n] \mid f_{m_1}(l) \neq f_{m_2}(l)\}$. For $k_1,k_2\in[n]$, the set $I:=[k_1, k_2] \subset [n]$ is called a maximal interval (of the discrepancy) between $f_{m_1}$ and $f_{m_2}$, if $I$ is a nonempty set such that $I\subset D(m_1,m_2)$, $k_1-1 \notin D(m_1,m_2)$, and $k_2+1 \notin D(m_1,m_2)$.}
	\end{defn}
	
	Due to Item~\ref{ite:def:1:3} of Definition~\ref{def:1}, the maximal intervals between $f_m,f_{m^*}$ exhibit symmetry about $n/2$, as shown in the next proposition. 
	
	\begin{prop}\label{prop:sym-single}
		Let $[k_1,k_2]\subset\llb n\rrb$ be a maximal interval between $f_{m}$ and $f_{m^*}$. If $k_2 + 1 < n-k_2$, i.e., $k_2 < \lfloor n/2\rfloor$, then $[n-k_2,n-k_1]$ is another maximal interval between $f_{m}$ and $f_{m^*}$. Similarly, if $k_1 > \lceil n/2 \rceil $, then $[n-k_2,n-k_1]$ is another maximal interval between $f_{m}$ and $f_{m^*}$. If $k_1 \leq \lceil n/2 \rceil$ and $k_2 \geq \lfloor n/2\rfloor$, then it is necessary that $k_2=n-k_1$ and so $k_1\leq \lfloor n/2\rfloor$ and $k_2\geq \lceil n/2\rceil$.
	\end{prop}
	\begin{proof}
		Since $f_m(l)\neq f_{m^*}(l)$ for $l\in [k_1,k_2]$, by Item~\ref{ite:def:1:3} of Definition~\ref{def:1}, we have $f_m(l)\neq f_{m^*}(l)$ for $l\in [n-k_2,n-k_1]$. The proposition follows by inspecting the intersection of $[k_1,k_2]$ and $[n-k_2,n-k_1]$.
	\end{proof}
	
	As the focus of this paper is on constant-weight strings, let us mention the following simple observation for constant-weight CWFs, which is also a consequence of Item~\ref{ite:def:1:3} of Definition~\ref{def:1}.
	\begin{prop}\label{symInterval}
		Assume $f$ is a constant-weight CWF. Let $m_1,m_2\in [2h]$ and $k_1, k_2\in [n-1]$.
		{If $[k_1, k_2] \subset [n-1]$} is a maximal interval between $f_{m_1}$ and $f_{m_2}$, then $[n-k_2,n-k_1]$ is a maximal interval between $f_{{m}^*_1}$ and $f_{{m}^*_2}$.
	\end{prop}
	
	Next, we introduce the notion of swap between functions $f_{m_1}$ and $f_{m_2}$ for any $m_1,m_2 \in [2h]$ {that ensures the resulting component functions still form a CWF}. In view of Proposition~\ref{prop:sym-single}, the swap operation has to be defined properly so that the symmetry of $f_{m_1},f_{m_1^*}$ and $f_{m_2},f_{m_2^*}$ are preserved after swapping, i.e., Item~\ref{ite:def:1:3} of Definition~\ref{def:1} is still satisfied by the new functions obtained after swapping.
	
	\begin{defn}\label{def:swap}
		Let $f$ be a CWF, $m_1, m_2 \in [2h]$, and $I\subset  \llb n \rrb$ be a maximal interval between $f_{m_1}$ and $f_{m_2}$.
		Let $g$ be the CWF obtained from $f$ by swapping the image of $(l,m_1)$ under $f$ for that of $(l,m_2)$, and the image of $(n-l,m_1^*)$ under $f$ for that of $(n-l,m_2^*)$ for all $l \in I$.  More precisely,
		if $m_1^*\neq m_2$, then 
		$g$ satisfies $g_m =  f_m$ for $m\in [2h]\setminus \{m_1,m_1^*,m_2,m_2^*\}$ and
		\begin{equation*}
			\begin{aligned}
				g_{m_1}(l) &=
				\begin{cases}
					f_{m_2}(l), & l \in I\\
					f_{m_1}(l), & l \notin I
				\end{cases},
				& g_{m^*_1}(n-l) &=
				\begin{cases}
					f_{m^*_2}(n-l) , & l \in I\\
					f_{m^*_1}(n-l) , & l \notin I
				\end{cases},\\
				g_{m_2}(l) &=
				\begin{cases}
					f_{m_1}(l), & l \in I\\
					f_{m_2}(l), & l \notin I
				\end{cases},
				& g_{m^*_2}(n-l) &=
				\begin{cases}
					f_{m^*_1}(n-l) , & l \in I\\
					f_{m^*_2}(n-l) , & l \notin I
				\end{cases}.
			\end{aligned}
		\end{equation*}
		If $m_1^*=m_2$, then 
		$g$ satisfies $g_m =  f_m$ for $m\in [2h]\setminus \{m_1,m_1^*\}$ and writing $I=[k_1,k_2],\bar{I}:=[n-k_2,n-k_1]$, we have
		\begin{equation*}
			\begin{aligned}
				g_{m_1}(l) &=
				\begin{cases}
					f_{m_2}(l), & l \in I\cup \bar{I}\\
					f_{m_1}(l), & l \notin I\cup \bar{I}
				\end{cases},
				& g_{m^*_1}(n-l) &=
				\begin{cases}
					f_{m^*_2}(n-l) , & l \in I\cup \bar{I}\\
					f_{m^*_1}(n-l) , & l \notin I\cup \bar{I}
				\end{cases}.
			\end{aligned}
		\end{equation*}
		Denote the mapping $(f, I, m_1, m_2) \mapsto  g$ by $\phi$.	
	\end{defn}
	
	
	The ideas of maximal intervals and swapping are particularly helpful in establishing the necessity of the conditions for unique reconstruction as we will see in Section~\ref{sec:necessity}.
	
	If $f$ is constant-weight, then by the rotational symmetry of $f_m$ and $f_{m^*}$, the behavior of $f$ completely characterized by $(l,f_m(l))$ for all $l\leq \lfloor n/2\rfloor$ and $m\in[2h]$. This motivates us to look at the ``median weight'' of the component functions $\{f_m\}$, introduced below. 
	\begin{defn}\label{def:med}
		For $m\in [2h]$, the \emph{median weight} of $f_m$ is defined to be $\med(f_m)=\frac{1}{2} \big(f_m(\lfloor n/2 \rfloor) + f_m(\lceil n/2 \rceil) \big)\in\mathbb{R}$. For $w\in\mathbb{R}$, let ${{A}_f(w)} = \{m\in[2h] \mid \med(f_m) = w \}$ be the set of indices of the component functions $\{f_m\}$ for which the median weight is $w$. If $f$ is clear from the context, denote {${A}_f(w)$ by ${A}(w)$} for simplicity.
	\end{defn}
	
	The set $A(w)$ plays an important role in showing the sufficiency of the conditions for unique reconstruction in Section~\ref{sec:sufficiency}. Note that if $f$ has constant weight $\bar{w}$ then $A_f(\bar{w}/2)$ must be even. In fact, for any $m\in [2h]$, we have $m\in A(\bar{w}/2)$ if and only if $m^*\in A(\bar{w}/2)$, due to the $180$-degree rotational symmetry of $f_m$ and $f_{m^*}$ about $(n/2,\bar{w}/2)$. 
	
	As mentioned earlier, CWFs may be used as a proxy for reconstructing multisets given $M$. In fact, our reconstruction algorithms presented in Section~\ref{sec:algorithm} essentially reconstruct CWFs whose corresponding multisets are compatible with $M$, and such CWFs are said to be ``solutions'' to $M$. 
	
	\begin{defn}\label{def:3}
		A CWF $f \colon T \to \llb n\rrb$ is called a \emph{solution} to the composition multiset $M$ if the multiset equality {$M=\{(l-f_m(l),f_m(l)) \mid m \in [2h],l \in [n] \}$ holds}.
	\end{defn}
	
	\begin{remark}\label{re:swap}
		{If $f$ is a solution to $M$ and $I\subset\llb n\rrb$ is a maximal interval between $f_{m_1}$ and $f_{m_2}$, then $g=\phi( f, I, m_1, m_2)$ is also a solution to $M$.}
	\end{remark}
	
	In order to recover all multisets of strings compatible with $M$, it suffices to find all CWF solutions to $M$. Therefore, it is helpful to establish connections between multiset $M$ and CWF $f$, which is what we will do next.
	
	\begin{defn}\label{def:5}
		Let $f\colon T\to \llb n\rrb$ be a CWF. For $(l,w)\in \llb n\rrb^2$, let $A_f(l,w)=\{m\in [2h]\mid f_m(l)=w\}$. When the underlying CWF $f$ is clear from the context, denote {${A}_f(l,w)$ by ${A}(l,w)$} for simplicity.
	\end{defn}
	
	\begin{defn}\label{def:num-a}
		For $(l,w)\in \llb n\rrb^2$, let $a_{l,w}$ be the number of pairs $(l-w,w)$ in {$M$} if $(l,w)\neq (0,0)$ and define $a_{0,0}=2h$. 
	\end{defn}
	
	\begin{remark}\label{re:M-f}
		By Definition~\ref{def:5} and \ref{def:num-a}, $|A(l,w)|$ is the number of functions in $\{f_m\}$ that satisfies $f_m(l)=w$, {and $a_{l,w}$ is the number of length-$l$ prefixes and suffixes of weight $w$.} Therefore, by Definition~\ref{def:3}, {a CWF $f$ is a solution to $M$ if and only  $|A(l,w)|=a_{l,w}$ for all $(l,w)\in\llb n\rrb^2$.} 
	\end{remark}

	By Remark~\ref{re:M-f}, to find a solution $f$ to $M$, one may plot the elements of the multiset $M$ on the two dimensional grid $T$ and construct a CWF $f$ such that it passes the point $(l,w)$ exactly $a_{l,w}$ times on the grid.
	Below we mention a few basic properties of $A(l,w)$ and $a_{l,w}$ that immediately {follow} from Definition~\ref{def:5} and \ref{def:num-a}.
	
	\begin{prop}\label{prop:p}
		\begin{enumerate}[label=(\roman*)]
			\item For $l\in [n]$ and $w_1,w_2\in \llb n\rrb$, if $w_1 \neq w_2$ then $A(l,w_1) \cap A(l,w_2) = \emptyset$.\label{prop:p:ite:1}
			\item For $(l,w)\in\llb n-1\rrb^2$, it holds that $A(l,w)\subset A(l+1,w)\cup A(l+1,w+1)$.\label{prop:p:ite:2}
			\item For $(l,w)\in [n]^2$, it holds that $A(l,w)\subset A(l-1,w)\cup A(l-1,w-1)$.\label{prop:p:ite:3}
		\end{enumerate}
	\end{prop}
	
	Note that \ref{prop:p:ite:2} (resp., \ref{prop:p:ite:3}) of Proposition~\ref{prop:p} simply says the weight of a substring cannot decrease (resp., increase) if its length increases (resp., decreases).
	As mentioned previously, a solution $f$ to $M$ must pass $(l,w)$ exactly $a_{l,w}$ times. To further assist in finding such CWFs, we will be interested in the number of {length-$l$ weight-$w$} prefixes and suffixes whose weight {remains} the same if the length {decreases}, and the number of those whose weight deceases with the length. {They are denoted by $b_{l,w}$ and $c_{l,w}$, introduced in the next definition.} 
	
	\begin{defn}\label{def:3.5}
		Let $f$ be a solution to $M$.  For all $(l,w) \in [n]^2$, define $b_{l,w} = |A(l,w) \cap A(l-1,w)|$ and $c_{l,w} = |A(l,w) \cap A(l-1,w-1)|$. {Moreover, define $b_{l,0} = |A(l,0)|$, $c_{l,0} = 0$ for all $l \in [n]$.}
	\end{defn}
	
	\begin{prop}\label{prop:3.7}
		Let $f$ be a solution to $M$. {The numbers $\{b_{l,w},c_{l,w} \mid (l,w) \in [n] \times \llb n \rrb; w \leq l \}$ can be computed from the numbers $\{a_{l,w} \mid (l,w)\in\llb n\rrb^2; w \leq l \}$.}
	\end{prop}
	\begin{proof}
		Since $f$ is a solution to $M$, we have $a_{l,w}=|A(l,w)|$. By Definition~\ref{def:3.5}, $b_{l,l}=0$, $c_{l,l}=a_{l,l}$ for all $l\in [n]$. It remains to find $b_{l,w},c_{l,w}$ where $0\leq w\leq l-1$. By \ref{prop:p:ite:1} of Proposition~\ref{prop:p}, $A(l-1,w)$ and $A(l-1,w-1)$ are disjoint. Therefore, $b_{l,w}+c_{l,w}\leq a_{l,w}$. At the same time, by \ref{prop:p:ite:3} of Proposition~\ref{prop:p}, we have $a_{l,w}\leq b_{l,w}+c_{l,w}$, and thus
		\begin{align}
			a_{l,w}= b_{l,w}+c_{l,w},\quad (l,w)\in [n] \times \llb n \rrb.\label{eq:sum1}
		\end{align}
		Using \ref{prop:p:ite:2} of Proposition~\ref{prop:p}, we have $a_{l,w}\leq b_{l+1,w}+c_{l+1,w+1}$. By \ref{prop:p:ite:1} of Proposition~\ref{prop:p}, $A(l+1,w)$ and $A(l+1,w+1)$ are disjoint so we also have ${b_{l+1,w}+c_{l+1,w+1}\leq a_{l,w}}$.
		Therefore,
		\begin{align}
			a_{l,w}= b_{l+1,w}+c_{l+1,w+1},\quad (l,w)\in \llb n-1\rrb^2.\label{eq:sum2}
		\end{align}
		It follows from \eqref{eq:sum2} that 
		$b_{l,l-1}=a_{l-1,l-1}-c_{l,l}$ for $l\in[n]$. Using \eqref{eq:sum1}, we obtain $c_{l,l-1}=a_{l,l-1}-b_{l,l-1}$ for $l\in[n]$. So we have found $\{b_{l,l-1},c_{l,l-1}\mid l\in[n]\}$. Next, from \eqref{eq:sum2} and \eqref{eq:sum1} we have $b_{l,l-2}=a_{l-1,l-2}-c_{l,l-1}$ and $c_{l,l-2}=a_{l,l-2}-b_{l,l-2}$ for $l\in[2,n]$. Repeating this process, we can determine $\{b_{l,l-i},c_{l, l-i}\mid l\in[i,n]\}$ for all $i\in[n]$.
	\end{proof}
	
	\begin{remark}\label{re:bc}
		As a consequence of Proposition~\ref{prop:3.7}, the numbers $\{b_{l,w}\}$ and $\{c_{l,w}\}$ can be found by inspecting $M$, and thus they are properties of $M$ in the sense that all solutions to $M$ result in the same $\{b_{l,w}\}$ and $\{c_{l,w}\}$. In fact, from the recursive procedures in the above proof, for $w\leq l$ we have 
		\begin{align*}
			b_{l,w} &= \sum_{v=w}^{l-1} a_{l-1,v}-\sum_{v=w+1}^{l} a_{l,v},\\
			c_{l, w} &= \sum_{v=w}^{l} a_{l,v} - \sum_{v=w}^{l-1} a_{l-1,v}.
		\end{align*}
		Since $b_{l,w}\geq 0, c_{l, w}\geq 0$, it follows that for all $w\leq l $
		\begin{align}
			\sum_{v=w}^{l-1} a_{l-1,v}&\geq \sum_{v=w+1}^{l} a_{l,v}, \label{eq:horizontal}\\
			\sum_{v=w}^{l} a_{l,v} &\geq \sum_{v=w}^{l-1} a_{l-1,v}. \label{eq:vertical}
		\end{align}		
	\end{remark}
	
	The numbers $\{a_{l,w}\},\{b_{l,w}\},\{c_{l,w}\}$ are instrumental in analyzing the possible behaviors of the component functions $\{f_m\}$ in Section~\ref{sec:algorithm}.
	
	\section{Necessary and sufficient conditions for unique reconstruction}\label{sec:conditions}
	
	In this section, we assume $H$ is a multiset of $h$ strings of length $n$ and weight $\bar{w}$. The main result of this section is stated in the following theorem.
	
	\begin{theorem}\label{UniqueThm}
		Let $f$ be a solution to $M$. 
		{There} is exactly one multiset of strings (up to reversal) compatible with $M$, i.e., $|\mathcal{H}|=1$, if and only if {$f$ satisfies the following conditions}:
		\begin{enumerate}[label=(\roman*)]
			\item \label{thm:unique:ite:1}
			For any $m_1, m_2 \in [2h]$ with $m_1^* = m_2$, there exist at most two maximal intervals between $f_{m_1}$ and $f_{m_2}$.
			
			\item \label{thm:unique:ite:2}
			For any $m_1, m_2 \in [2h]$ with $m_1^* \neq m_2$, there exists at most one maximal interval between $f_{m_1}$ and $f_{m_2}$.
		\end{enumerate}
	\end{theorem}   
	
	\subsection{Necessity}\label{sec:necessity}
	
	To give a rough idea of why the conditions in Theorem~\ref{UniqueThm} are necessary for unique reconstruction, let first consider some simple examples for the case where there is a single string $\bm{t}$. Suppose $\bm{t}=011101$ and so $\cev{\bm{t}}=101110$. A string $\bm{s}=101110$, which has the same prefix-suffix compositions as $\bm{t}$, can be obtained by swapping the first two and last two bits of $\bm{t}$ for those of $\cev{\bm{t}}$. Note that $\bm{s}$ is simply $\cev{\bm{t}}$ and we only obtain the reversal of $\bm{t}$ after swapping. Using the language of CWFs, let $f$ be the CWF induced by $\{\bm{t}\}$ and $f_1,f_2$ are the functions corresponding to $\bm{t},\cev{\bm{t}}$. We observe that there are only two maximal intervals between $f_1$ and $f_2$. 
	
	Next, let us examine an example where we produce a new string by swapping. Take $\bm{t}=010101$ and so $\cev{\bm{t}}=101010$. In this case, there are three maximal intervals between the corresponding functions $f_1,f_2$. Swapping the first two and last two bits of $\bm{t}$ with $\cev{\bm{t}}$, we obtain a new string $\bm{s}=100110$. Clearly, $\bm{s}\neq \cev{\bm{t}}$ and $\bm{s}$ has the same prefix-suffix compositions as $\bm{t}$.
	
	From the above two examples, one may expect that if there are at least three maximal intervals between $\bm{t}$ and $\cev{\bm{t}}$ then $\bm{t}$ cannot be uniquely reconstructed, and therefore, the existence of at most two maximal intervals is necessary for unique reconstruction of a single string. A similar analysis can also be carried out for two strings that are not reversal of each other, and it turns out that the existence of at most one maximal interval is necessary for unique reconstruction in this case. 
	
	\begin{lemma}\label{le:same-str}
		Let $f$ be a solution to $M$ and let $m \in [2h]$. 
		If there exists at least three maximal intervals between $f_{m}$ and $f_{m^*}$, then $|\mathcal{H}| > 1$
	\end{lemma}
	\begin{proof}
		Let $I_1=[k_1,k_2],I_2=[k_3,k_4],I_3=[k_5,k_6]$ be three maximal intervals between $f_{m}$ and $f_{m^*}$. Without loss of generality, we may assume $0<k_1\leq k_2 < k_3\leq k_4 < k_5\leq k_6 <n$. Construct $g = \phi(f, I_1, m, m^*)$. By Remark~\ref{re:swap}, $g$ is also a solution to $M$. Let $\bar{I}_1=[n-k_2,n-k_1]$.
		By construction of $g$, we have $g_{m}\neq f_{m}$ on $I_1\cup\bar{I}_1$. Moreover, $g_{m}\neq f_{m^*}$ on either $I_2$ or $I_3$ since $\bar{I}_1$ cannot equal both of them. Therefore, the string corresponding to $g_m$ is not same as $f_m,f_{m^*}$ and we have $[H_g]\neq [H_f]$.
		Hence, if there exists at least three maximal intervals between $f_{m}$ and $f_{m^*}$, then $|\mathcal{H}| > 1$.
	\end{proof}
	
	\begin{lemma}\label{notUnique}
		Let $f$ be a solution to $M$ and let $m_1, m_2 \in [2h]$ with $m_1^* \neq m_2$. If there exist at least two maximal intervals between $f_{m_1}$ and $f_{m_2}$, then $|\mathcal{H}| > 1$.
	\end{lemma}
	\begin{proof}
		Let $I_1$, $I_2$ be two maximal intervals between $f_{m_1}$ and $f_{m_2}$. Without loss of generality, assume $\{\lfloor n/2 \rfloor, \lceil n/2 \rceil\} \not\subset I_1$. Construct $g = \phi(f, I_1, m_1, m_2)$. By Remark~\ref{re:swap}, $g$ is also a solution to $M$. In the following, we will show that
		\begin{align}
			\{f_{m_1},f_{m_1^*},f_{m_2},f_{m_2^*}\} \neq \{g_{m_1},g_{m_1^*},g_{m_2},g_{m_2^*}\},\label{eq:diff}
		\end{align}
		implying  $|\mathcal{H}|>1$. 
		
		Since $m_1^*\neq m_2$ and $I_1$ is a maximal interval between $f_{m_1}$ and $f_{m_2}$, by construction of $g$, we have $g_{m_1}\neq f_{m_1}$ and $I_1$ is the only maximal interval between $g_{m_1}$ and $f_{m_1}$. We claim $g_{m_1}\neq f_{m_1^*}$ also holds. Indeed, if $g_{m_1}=f_{m_1^*}$, then $I_1$ is the only maximal interval between $f_{m_1^*}$ and $f_{m_1}$. However, since $ I_1\not\supset \{\lfloor \frac{n}{2} \rfloor, \lceil \frac{n}{2} \rceil\}$, according to Proposition~\ref{prop:sym-single}, there are at least two maximal intervals between $f_{m_1^*}$ and $f_{m_1}$, which is a contradiction. Therefore, $g_{m_1}\neq f_{m^*_1}$. So far we have shown
		\begin{align}
			g_{m_1} \neq  f_{m_1},\quad  g_{m_1} \neq  f_{m^*_1}. \label{eq6}
		\end{align}
		By construction of $g$, we have $g_{m_1} \neq f_{m_2}$. If $g_{m_1} \neq f_{m^*_2}$, then \eqref{eq:diff} holds and we are done. 
		
		Consider the case where $g_{m_1} = f_{m^*_2}$. Using arguments similar to those leading to \eqref{eq6}, one can obtain 
		\begin{align}
			g_{m_2} \neq  f_{m_2},\quad  g_{m_2} \neq  f_{m^*_2}. \nonumber
		\end{align}
		By construction of $g$, we also have $g_{m_2} \neq  f_{m_1}$. Next, we would like to show  $g_{m_2}\neq f_{m_1^*}$ for \eqref{eq:diff} to hold. Recall that the set $I_1$ is the only maximal interval between $g_{m_1}$ and $f_{m_1}$. Since $g_{m_1}=f_{m_2^*}$, it follows that $I_1$ is the only maximal interval between $f_{m_2^*}$ and $f_{m_1}$.  Write $I_1 = [k_1, k_2]$. By Proposition~\ref{symInterval}, the set $[n-k_2, n-k_1]$ is a maximal interval between $f_{m_2}$ and $f_{m^*_1}$, and so $f_{m_2}(n-k_1)\neq f_{m^*_1}(n-k_1)$. Since $ I_1\not\supset \{\lfloor \frac{n}{2} \rfloor, \lceil \frac{n}{2} \rceil\}$, we have $I_1\cap [n-k_2,n-k_1]=\emptyset$. By construction of $g$, we have $g_{m_2}(n-k_1)=f_{m_2}(n-k_1)$ and it follows that $g_{m_2}(n-k_1)\neq f_{m^*_1}(n-k_1)$, i.e., $g_{m_2}\neq f_{m_1^*}$. Therefore, \eqref{eq:diff} also holds.
		
		In summary, no matter whether $g_{m_1}$ and $f_{m_2^*}$ are the same or not, \eqref{eq:diff} always holds. It follows that the multisets corresponding to $f,g$ satisfy $[H_f]\neq [H_g]$, and thus $|\mathcal{H}|>1$.
	\end{proof}
	
	The necessity part of Theorem~\ref{UniqueThm} follows from Lemma~\ref{le:same-str} and \ref{notUnique}.

	\subsection{Sufficiency}\label{sec:sufficiency}
	From the above discussion on the necessity, it is not difficulty to see that if $f$ is solution to $M$ such that the conditions in Theorem~\ref{UniqueThm} hold, then any CWF $g$ resulted from a series of the swap operations between $f_1,\ldots,f_{2h}$ satisfies $[H_g] = [H_f]$. Therefore, the sufficiency of the conditions follows if one can further show that any solution to $M$ can be obtained from repeated applications of the swap operation between $f_1,\ldots,f_{2h}$. However, it is in general not obvious to establish such a connection between $f$ and an arbitrary solution to $M$. Thus, we take a different approach to showing the sufficiency.
	{Our main idea is to translate the conditions in Theorem~\ref{UniqueThm} to properties shared by all solutions to $M$ and utilize these properties to establish the sufficiency of the conditions.}

	As mentioned before, the CWF $f$ induced by $h$ strings of length $n$ and weight $\bar{w}$ is determined by the behaviors of the functions $\{f_m\}$ on $\llb\lfloor n/2\rfloor\rrb$ because of the constant weight. Based on the values that the functions $\{f_m\}$ take at $n/2$, i.e., the median weight $\med(f_m)$, the functions $\{f_m\}$ can be formed into groups $A(w),w=0,1/2,1,\ldots,\bar{w}$. In the following, we analyze the behaviors of the functions $\{f_m\}$ according to their membership in these groups. Let us first rephrase the conditions for $f_m,f_{m^*}$ in Theorem~\ref{UniqueThm} using their rotational symmetry.
	
	\begin{prop}\label{prop:f-dual}
		Let $f$ be a solution that satisfies the conditions in Theorem~\ref{UniqueThm}. Then the following holds:
		\begin{enumerate}[label=(\roman*)]
			\item For any $m\in A(\bar{w}/2)$, either $f_m=f_{m^*}$ or there are exactly two maximal intervals between $f_m$ and $f_{m^*}$ and exactly one of the two intervals is contained in $[\lfloor n /2\rfloor ]$.
			\item For any $m\in [2h]\setminus A(\bar{w}/2)$, there is exactly one maximal interval between $f_m$ and $f_{m^*}$.
		\end{enumerate}
	\end{prop}
	\begin{proof}
		As mentioned previously, for any $m\in [2h]$, we have $m\in A(\bar{w}/2)$ if and only if $m^*\in A(\bar{w}/2)$. For any $m\in A(\bar{w}/2)$ with $f_m\neq f_{m^*}$, since $f$ satisfies the conditions in Theorem~\ref{UniqueThm}, there is either one maximal interval or two maximal intervals between $f_m$ and $f_{m^*}$. Since $\med(f_m)=\med(f_{m^*})$, it follows that at least one of the maximal intervals is contained in $[\lfloor n/2 \rfloor]$ or $[\lfloor n/2 \rfloor+1, n]$. Suppose there is only one maximal interval between $f_m$ and $f_{m^*}$. Then the maximal interval is contained in $[\lfloor n/2 \rfloor]$ or $[\lfloor n/2 \rfloor+1, n]$, but by Proposition~\ref{prop:sym-single}, there are two maximal intervals between $f_m$ and $f_{m^*}$, which is a contradiction. So there are exactly two maximal intervals. Now suppose the two intervals are both in $[\lfloor n/2 \rfloor]$ or both in $[\lfloor n/2 \rfloor+1, n]$. Then by Proposition~\ref{prop:sym-single}, there are more than two maximal intervals between $f_m$ and $f_{m^*}$, which is contradiction. It follows that exactly one of the two intervals is contained in $[\lfloor n /2\rfloor ]$.
		
		For any $m\in [2h]\setminus A(\bar{w}/2)$, we have $\med(f_m)\neq \med(f_{m^*})$ so by Proposition~\ref{prop:sym-single} there exists one maximal interval between $f_m$ and $f_{m^*}$ that contains $\{\lfloor n/2\rfloor,\lceil n/2\rceil\}$. Furthermore, if there is another maximal interval contained in $[\lfloor n/2 \rfloor]$ or $[\lfloor n/2 \rfloor+1, n]$, by Proposition~\ref{prop:sym-single} there are at least three maximal intervals between $f_{m}$ and $f_{m^*}$, which is a contradiction to the conditions in Theorem~\ref{UniqueThm}. Therefore, there is exactly one maximal interval between $f_m$ and $f_{m^*}$.
	\end{proof}
	
	Below we introduce two more definitions that are helpful for discussing the behaviors of the functions $\{f_m\}$ in this subsection.
	
	\begin{defn}
		For $m \in [2h]$ and $I \subset \llb n\rrb$, let $\mathcal{G}(f_m,I):=\{(l,f_m(l))\mid l\in I \}$ be the graph of $f_m$ over $I$ and denote $\mathcal{G}(f_m):=\mathcal{G}(f_m,\llb n\rrb)$. 
	\end{defn}
	\begin{defn}\label{def:pts}
		{An element $(l,w) \in \llb n\rrb^2$ is called a branching point if $b_{l,w} > 0$ and $c_{l,w} > 0$. An element $(l,w) \in \llb n \rrb^2$ is called a merging point if $b_{l+1,w} > 0$ and $c_{l+1,w+1} > 0$.}\footnote{The branching and merging points are so named because we would like to visualize the graphs $\{\mathcal{G}(f_m) \mid m \in [2h] \}$ evolving from $l = n$ to $l=0$.}
	\end{defn}
	
	Using Proposition~\ref{prop:f-dual}, we examine the conditions in Theorem~\ref{UniqueThm} in terms of the branching points and merging points on $\{f_m\}$ in a series of lemmas below. Lemma~\ref{branch-merge} first examines the functions $\{f_m\}$ for which $m\in [2h]\setminus A(\bar{w}/2)$.
	
	\begin{lemma}\label{branch-merge}
		{Let $f$ be a solution to $M$ that satisfies the conditions in Theorem~\ref{UniqueThm} and let $m \in [2h] \setminus A(\bar w /2)$. If $(l, w) \in \mathcal{G}(f_{m})$ is a merging point, then there are no branching points in $\mathcal{G}(f_m, [l])$.}
	\end{lemma}
	
	\begin{proof}
		If $(l, w) \in \mathcal{G}(f_{m})$ is a merging point, there exists $m_1\in[2h]\setminus\{ m\}$ such that $f_m(l+1) \neq f_{m_1}(l+1)$ and $f_m(l) = f_{m_1}(l)$.
		We claim $\mathcal{G}(f_{m}, [l]) = \mathcal{G}(f_{m_1}, [l])$. Indeed, if $\mathcal{G}(f_{m}, [l]) \neq \mathcal{G}(f_{m_1}, [l])$ then there is at least one maximal interval between $f_{m}$ and $f_{m_1}$ contained in $[l-1]$, in addition to the one contained in $[l+1,n]$. Since $f$ satisfies the conditions in Theorem~\ref{UniqueThm}, we must have $m_1=m^*$. However, by Proposition~\ref{prop:f-dual}, if $m_1=m^*$ there should be only one maximal interval between $f_{m_1}$ and $f_{m}$, leading to a contradiction. Hence, $\mathcal{G}(f_{m}, [l]) = \mathcal{G}(f_{m_1}, [l])$.
		
		Suppose $(k,v) \in \mathcal{G}(f_{m},[l])$ is a branching point. Then there exists $m_2 \in [2h] \setminus \{m, m_1\}$ such that $f_{m_2}(k)=f_m(k)$ and $f_{m_2}(k-1)\neq f_m(k-1)$. Since $(l, w) \in \mathcal{G}(f_{m})$ is a merging point and we have $f_{m}(l+1) \neq f_{m_1}(l+1),\mathcal{G}(f_{m}, [l]) = \mathcal{G}(f_{m_1}, [l])$, there must exist $\tilde{m} \in \{m, m_1\}$ such that $f_{\tilde{m}}(k-1)\neq f_{m_2}(k-1)$ and $f_{\tilde{m}}(l+1)\neq f_{m_2}(l+1)$. It follows that there are two maximal intervals between $f_{\tilde{m}}$ and $f_{m_2}$, and therefore, by the conditions in Theorem~\ref{UniqueThm} we have $m_2=\tilde{m}^*$. 
		
		If $\tilde{m}\in [2h]\setminus A(\bar{w}/2)$, then by Proposition~\ref{prop:f-dual} there should be exactly one maximal interval between $f_{\tilde{m}},f_{m_2}$, which is a contradiction. 
		
		If $\tilde{m}\in A(\bar{w}/2)$, then $\tilde{m}=m_1$ and $m_2=m_1^*\in A(\bar{w}/2)$. 
		So the median weights of $f_m$ is different from that of $f_{m_1},f_{m_2}$ and there exists $l_*\in \{\lfloor n/2\rfloor,\lceil n/2\rceil \}$ such that $f_m(l_*)\neq f_{m_1}(l_*),f_m(l_*)\neq f_{m_2}(l_*)$. 
		Since $\mathcal{G}(f_{m}, [l]) = \mathcal{G}(f_{m_1}, [l])$, we have $l< l_*$. It follows that $k< l_*$. So there exists a maximal interval between $f_m,f_{m_2}$ that is contained in $[k,n]$, in addition to the one contained in $[k-1]$. Since $m \in [2h] \setminus A(\bar w /2)$ we have $m_2\neq m^*$, and therefore, by the conditions in Theorem~\ref{UniqueThm} there should be only one maximal intervals between $f_m,f_{m_2}$, which is a contradiction.
		
		Thus, there are no branching points in $\mathcal{G}( f_{m},{[l]})$.
	\end{proof}

	\begin{remark}\label{re:symmetry}
		{One can also verify that if $(l,w) \in \mathcal{G}(f_{m})$ is a branching point, where $m\in[2h]\setminus A(\bar w /2)$}, then $\mathcal{G}(f_{m},{[l, n]})$ has no merging points.
	\end{remark}
	
	The next three lemmas examine the behaviors of $\{f_m\}$ for which $m\in A(\bar{w}/2)$. In particular, the discussion is based on whether $f_m,m\in A(\bar{w}/2)$ are all the same or not.

	\begin{lemma}\label{le:half-same}
		Let $f$ be a solution to $M$ that satisfies the conditions in Theorem~\ref{UniqueThm}. If $f_{m},m\in A(\bar w /2)$ are all the same,
		then there are no branching points in $\mathcal{G}(f_m,[\lfloor n/2\rfloor ])$ for all $m \in A(\bar w /2)$.  
	\end{lemma}
	\begin{proof}
		Suppose there exist branching points in $\mathcal{G}(f_m,[\lfloor n/2\rfloor ])$ for some $m\in A(\bar w /2)$ and let $(l_*,w_*)\in \mathcal{G}(f_m,[\lfloor n/2\rfloor ])$ be a branching point. Since $f_{m_1}=f_{m_2}$ for any $m_1,m_2 \in A(\bar w /2)$, there must exist $\tilde{m}\in [2h]\setminus A(\bar w /2)$ such that $f_{\tilde{m}}(l_*)=f_{m}(l_*)$ and $f_{\tilde{m}}(l_*-1)\neq f_{m}(l_*-1)$. Moreover, since $\tilde{m}\notin A(\bar w/ 2)$, we have $f_{\tilde{m}}(l)\neq f_{m}(l)$ for some $l\in [l_*, \lceil n/2\rceil ]$. It follows that there exist two maximal intervals between $f_m$ and $f_{\tilde{m}}$. This is a contradiction to Item~\ref{thm:unique:ite:2} in Theorem~\ref{UniqueThm} by noticing $m^*\neq \tilde{m}$ since $m^*\in A(\bar w/2)$ for all $m\in A(\bar w/2)$.
	\end{proof}
	
	
	If $f_m,m\in A(\bar{w}/2)$ are not all the same, Lemma~\ref{le:half-diff-0} shows that graphs of $f_m$ over $[\lfloor n/2 \rfloor]$ are essentially of two kinds. The proof of Lemma~\ref{le:half-diff-0} is presented in Appendix~\ref{app:half-diff-0}.
	
	\begin{lemma}\label{le:half-diff-0}
		Let $f$ be a solution to $M$ that satisfies the conditions in Theorem~\ref{UniqueThm}. If $f_m,m \in A(\bar w /2)$ are not all the same, then there exists $m_1\in A(\bar{w}/2)$ such that there are exactly two maximal intervals between $f_{m_1}$ and $f_{m_1^*}$, and $f_m, m\in A(\bar{w}/2)\setminus\{m_1,m_1^*\}$ are all the same. Moreover, it holds that $\mathcal{G}(f_m,[\lfloor n/2 \rfloor])=\mathcal{G}(f_{m_1},[\lfloor n/2 \rfloor])$ for all $m\in A(\bar{w}/2)\setminus\{m_1,m_1^*\}$ or $\mathcal{G}(f_m,[\lfloor n/2 \rfloor])=\mathcal{G}(f_{m_1^*},[\lfloor n/2 \rfloor])$ for all $m\in A(\bar{w}/2)\setminus\{m_1,m_1^*\}$.
	\end{lemma}
	
	Using Lemma~\ref{le:half-diff-0}, we can further deduce the property of the branching points and merging points on $f_m,m\in A(\bar{w}/2)$.
	
	\begin{lemma}\label{le:half-diff}
		Let $f$ be a solution to $M$ that satisfies the conditions in Theorem~\ref{UniqueThm}.	
		If $f_m,m \in A(\bar w /2)$ are not all the same, then there exists $m_1 \in A(\bar w /2)$ such that there is a maximal interval $[l_1+1,l_2-1]\subset [n]$ between $f_{m_1}$ and $f_{m_1^*}$,
		where $l_2\leq \lfloor n/2 \rfloor$. Moreover, $(l_2,f_{m_1}(l_2))$ is the only branching point in $\mathcal{G}(f_{m}, [\lfloor n/2\rfloor ])$ and there is no merging point in $\mathcal{G}(f_{m}, [l_2, \lfloor n/2\rfloor ])$ for all $m\in A(\bar{w}/2)$. 
	\end{lemma}
	\begin{proof}
		By Lemma~\ref{le:half-diff-0} and Proposition~\ref{prop:f-dual}, there exists $m_1 \in A(\bar w /2)$ such that there is a maximal interval $[l_1+1,l_2-1]\subset [n]$ between $f_{m_1}$ and $f_{m_1^*}$,
		where $l_2\leq \lfloor n/2 \rfloor$. In addition, by Lemma~\ref{le:half-diff-0}, we have that $(l_1,f_{m_1}(l_1))$ is a merging point in $\mathcal{G}(f_m)$ for all $m\in A(\bar{w}/2)$. In what follows, let $m\in A(\bar{w}/2)$.
		
		Suppose there exists $l_*\in [\lfloor n/2\rfloor],l_*\neq l_2$ such that $(l_*,f_m(l_*))$ is a branching point. By Lemma~\ref{le:half-diff-0}, there exist $a\in [2h]\setminus A(\bar{w}/2)$ such that $f_a(l_*)=f_m(l_*)$ and $f_a(l_*-1)\neq f_m(l_*-1)$. It follows that there are two maximal intervals between $f_a,f_m$: one is contained in $[l_*-1]$ and {the other is contained in $[l_*+1, n]$ (since the median weight of $f_a$ is different from that of $f_m$)}. However, we have $a\neq m^*$, which is a contradiction to the conditions in Theorem~\ref{UniqueThm}. Therefore, $(l_2,f_{m_1}(l_2))$ is the only branching point in $\mathcal{G}(f_{m}, [\lfloor n/2\rfloor ])$.
		
		Suppose there exists $l_*\in[l_2, \lfloor n/2\rfloor ]$ such that $(l_*,f_m(l_*))$ is a merging point. By Lemma~\ref{le:half-diff-0}, there exist $a\in [2h]\setminus A(\bar{w}/2)$ such that $f_a(l_*)=f_m(l_*)$ and $f_a(l_*+1)\neq f_m(l_*+1)$. Since $(l_2,f_{m_1}(l_2))$ is the only branching point in $\mathcal{G}(f_{\tilde{m}}, [\lfloor n/2\rfloor ])$ for all $\tilde{m}\in A(\bar{w}/2)$, it follows from Lemma~\ref{le:half-diff-0} that there exists $b\in A(\bar{w}/2)$ such that $\mathcal{G}(f_b,[l_2])\neq \mathcal{G}(f_a,[l_2])$. Therefore, there exist two maximal interval between $f_a,f_b$: one contained in $[l_2]$ and the other is contained in $[l_*+1,n]$. However, we have $a\neq b^*$, which is a contradiction to the conditions in Theorem~\ref{UniqueThm}. Thus, there is no merging point in $\mathcal{G}(f_{m}, [l_2, \lfloor n/2\rfloor ])$.
	\end{proof}

	So far, we have translated the conditions in Theorem~\ref{UniqueThm} to properties of the branching points and merging points on $\{f_m\}$. The advantage of doing so is that properties of branching points and merging points are shared by all solutions to $M$.	
	Let $f,f'$ be two solutions to $M$. As a result of Remark~\ref{re:M-f} and \ref{re:bc}, $(l,w)\in [n]^2$ is a branching point in $\mathcal{G}(f_m)$ for some $m\in [2h]$ if and only if $(l,w)$ is a branching point in $\mathcal{G}(f'_{m'})$ for some $m'\in [2h]$. In particular, there is no branching point in $\mathcal{G}(f_m,I)$ for $I\subset \llb n\rrb$ if and only if there is no branching point in $\mathcal{G}(f'_{m'},I)$. The same statements hold for merging points.
	In view of this, we can then facilitate the description of the conditions in Theorem~\ref{UniqueThm} in terms of branching points and merging points to establish the sufficiency of the conditions in Theorem~\ref{UniqueThm}.
	
	Let us present two simple propositions that relate $f,f'$ using branching points and merging points.	
	\begin{prop}\label{prop:noBranch-0}
		{Let $f$, $f'$ be two solutions to $M$ and $[l_1, l_2] \subset [ n ]$. If there is no branching point in $\mathcal{G}(f_m,[l_1, l_2])$ and $f_m(l_2) = f'_{m'}(l_2)$ for some $m'\in[2h]$, then for any $l \in [l_1-1, l_2]$ it holds that {$f_m(l) = f'_{m'}(l)$}.} 
	\end{prop}
	
	\begin{proof}
		Suppose there exist $l \in [l_1-1, l_2]$ such that {$f_m(l) \neq f'_{m'}(l)$}. 
		Let $l_*\in [l_1, l_2]$ be such that $f_{m}(l_*)=f'_{m'}(l_*)$ and $f_{m}(l_*-1)\neq f'_{m'}(l_*-1)$. Let $w=f_m(l_*)$. Since $f,f'$ are solutions to $M$, it follows that the number of pairs $(l_*-w,w)$ in $M$ is at least $2$, i.e., $a_{l_*,w}\geq 2$. Moreover, we have $b_{l_*,w}\geq 1,c_{l_*,w}\geq 1$. Thus, by Definition~\ref{def:pts}, $(l_*,w)$ is a branching point in $\mathcal{G}(f_m,[l_1, l_2])$, resulting in a contradiction. 
	\end{proof}
	
	\begin{prop}\label{prop:noBranch}
		Let $f$, $f'$ be two solutions to $M$ and $[l_1, l_2] \subset [ n ]$. If $f_m(l_2) = f'_{m'}(l_2)$ and $f_m(l_1) \neq f'_{m'}(l_1)$, there must be a branching point in $\mathcal{G}(f_m,{[l_1+1, l_2]})$. Similarly, if $f_m(l_2) \neq f'_{m'}(l_2)$ and $f_m(l_1) = f'_{m'}(l_1)$, there must be a merging point in $\mathcal{G}(f_m,{[l_1, l_2-1]})$.
	\end{prop}
	\begin{proof}
		{The first part of the statement is a direct consequence of Proposition~\ref{prop:noBranch-0}. For the second part, we observe that there exists $l_* \in [l_1+1, l_2]$ such that $f_m(l_*) \neq f'_{m'}(l_*)$ and $f_m(l_*-1) = f'_{m'}(l_*-1)$. Without loss of generality, assume $f_m(l_*) = f_m(l_*-1) = w$ and $ f'_{m'}(l_*) = f'_{m'}(l_*-1) + 1$. These two equations imply $b_{l_*, w} \geq 1$ and $c_{l_*, w+1} \geq 1$, respectively. Thus, by Definition~\ref{def:pts}, $(l_*-1, w)$ is a merging point in  $\mathcal{G}(f_m,{[l_1, l_2-1]})$.}
	\end{proof}
	
	In the next two lemmas, we show that if $f,f'$ are two solutions to $M$ with $f$ satisfying the conditions in Theorem~\ref{UniqueThm}, then the multiset equality $\{f_m\mid m\in[2h]\}=\{f'_m\mid m\in[2h]\}$ must hold, thereby proving the sufficiency of the conditions in Theorem~\ref{UniqueThm} for unique reconstruction.

	\begin{lemma}\label{le:half-bar-w}
		Let $f$, $f'$ be two solutions to $M$ with $f$ satisfying the conditions in Theorem~\ref{UniqueThm}. 
		Let $\psi_1(f)=\{f_m \mid m\in [2h]\setminus A_f(\bar{w}/2) \}$ be a multiset and define $\psi_1(f')$ accordingly. Then $\psi_1(f) = \psi_1(f')$.
	\end{lemma}
	\begin{proof}
		Let $m\in S_f:=[2h]\setminus A_f(\bar{w}/2)$. Note that $f_m\neq f_{m^*}$ since their median weights are different. So there are branching points in $\mathcal{G}(f_m)$. Let $({l}_*,{w}_*)$ be the branching point in $\mathcal{G}(f_m)$ such that $l_*\leq l$ for any branching point $(l,w)\in\mathcal{G}(f_m)$. Let 
		\begin{align*}
			r=\begin{cases}
				0 & \text{if } w_* = {f_m(l_*-1)},\\
				1 & \text{if } w_* = {f_m(l_*-1)}+1.
			\end{cases}
		\end{align*} In other words, $r$ is an indicator of the behavior of $f_m$ to the left of the branching point $(l_*,w_*)$. By definition of $r$, we have $m\in S_f\cap A_{f}(l_*, w_*) \cap A_{f}(l_*-1, w_*-r)$. 
		
		Let $S_{f'}=[2h]\setminus A_{f'}(\bar{w}/2)$ and $m'\in S_{f'}\cap A_{f'}(l_*, w_*) \cap A_{f'}(l_*-1, w_*-r)$. In the following we will show $f'_{m'}=f_m$. 
		Since there is no branching point in $\mathcal{G}(f_m,[ l_*-1])$ and $f_m(l_*-1) = f'_{m'}(l_*-1)$, by Proposition~\ref{prop:noBranch-0} we have $f_m(l)=f'_{m'}(l)$ for any $l\in [0, l_*-1]$. 
		Suppose $f_m(l)\neq f'_{m'}(l)$ for some $l\in [l_*, n]$. Then by Proposition~\ref{prop:noBranch}, there is a merging point in $\mathcal{G}(f_{m},{[l_*, l-1]})$. But then by Lemma~\ref{branch-merge} there are no branching points in $\mathcal{G}(f_m,[l_*])$, contradicting that $(l_*,w_*)\in\mathcal{G}(f_m,[l_*])$ is a branching point. Thus, $f_m(l)= f'_{m'}(l)$ for all $l\in [l_*, n]$. It follows that $f_m=f'_{m'}$ for any $m'\in S_{f'}\cap A_{f'}(l_*, w_*) \cap A_{f'}(l_*-1, w_*-r)$.
		
		{{Next, let us show that $S_{f'}\cap A_{f'}(l_*, w_*) \cap A_{f'}(l_*-1, w_*-r)=A_{f'}(l_*, w_*) \cap A_{f'}(l_*-1, w_*-r)$. Toward a contradiction, suppose that} there exists $m_0 \in A_{f'}(\bar{w}/2) \cap A_{f'}(l_*, w_*) \cap A_{f'}(l_*-1, w_*-r)$. Then {$f'_{m_0}(l_*) = f_{m}(l_*)$ and $f'_{m_0}(l_*-1) = f_{m}(l_*-1)$}. Note that $\med(f'_{m_0}) \neq \med(f_m)$. {If $l_*-1\geq \lceil n/2 \rceil$,} by Proposition~\ref{prop:noBranch}, there must be a branching point in {$\mathcal{G}(f_m, [l_*-1])$, contradicting the assumption that} $l_*\leq l$ for any branching point $(l,w) \in \mathcal{G}(f_m)$. {If $l_*-1< \lceil n/2 \rceil$,} there must be a merging point in $\mathcal{G}(f_m, [l_*, \lceil n/2 \rceil])$, {but by Lemma~\ref{branch-merge} there should be no branching points in $\mathcal{G}(f_m, [l_*])$, contradicting that $(l_*,w_*)$ is a branching point}. We thus conclude $A_{f'}(\bar{w}/2) \cap A_{f'}(l_*, w_*) \cap A_{f'}(l_*-1, w_*-r) = \emptyset$ and so $S_{f'} \supset A_{f'}(l_*, w_*) \cap A_{f'}(l_*-1, w_*-r)$.}
		
		Note that for any $m'\in  S_{f'}\setminus( A_{f'}(l_*, w_*) \cap A_{f'}(l_*-1, w_*-r))$, we have $f'_{m'}\neq f_{m}$. So the multiplicity of $f_{m}$ in $\psi_1(f')$ is {$|S_{f'}\cap A_{f'}(l_*, w_*) \cap A_{f'}(l_*-1, w_*-r)|=| A_{f'}(l_*, w_*) \cap A_{f'}(l_*-1, w_*-r)|$}. {Taking $f'=f$, one can repeat the above arguments to show that $f_m=f_{\tilde{m}}$ for any $\tilde{m}\in S_{f}\cap A_{f}(l_*, w_*) \cap A_{f}(l_*-1, w_*-r)$ and the multiplicity of $f_m$ in $\psi_1(f)$ is $| A_{f}(l_*, w_*) \cap A_{f}(l_*-1, w_*-r)|$.}
		
		Since $f,f'$ are solutions to $M$, {$| A_{f}(l_*, w_*) \cap A_{f}(l_*-1, w_*-r)| = | A_{f'}(l_*, w_*) \cap A_{f'}(l_*-1, w_*-r)|$}, i.e., the multiplicity of $f_m$ in $\psi_1(f)$ equals the multiplicity of $f_{m}$ in $\psi_1(f')$. Furthermore, this holds for distinct $f_m\in\psi_1(f)$. Since $|S_f|=|S_{f'}|$, i.e., $|\psi_1(f)|=|\psi_1(f')|$, we obtain $\psi_1(f) = \psi_1(f')$.
	\end{proof}
	
	\begin{lemma}\label{le:half-bar-w-0}
		Let $f$, $f'$ be two solutions to $M$ with $f$ satisfying the conditions in Theorem~\ref{UniqueThm}. 
		Let $\psi_0(f)=\{f_m \mid m\in A_f(\bar{w}/2) \}$ be a multiset and define $\psi_0(f')$ accordingly. Then $\psi_0(f) = \psi_0(f')$.
	\end{lemma}
	
	The idea of the proof for Lemma~\ref{le:half-bar-w-0} is similar that for Lemma~\ref{le:half-bar-w}, whereas it relies on Lemma~\ref{le:half-same} and \ref{le:half-diff} instead of Lemma~\ref{branch-merge}. The complete proof is given in Appendix~\ref{app:half-bar-w-0}.
	
	It follows from Lemma~\ref{le:half-bar-w} and \ref{le:half-bar-w-0} that the conditions in Theorem~\ref{UniqueThm} are sufficient for unique reconstruction.
	
	\section{Reconstruction algorithms}\label{sec:algorithm}
	
	As before, we assume in this section that $M$ is the prefix-suffix compositions of the multiset $H$ of $h$ strings of length $n$ and weight $\bar{w}$. We present two algorithms that produce multisets of strings compatible with $M$. Both algorithms first construct CWFs and then find the corresponding multisets as in Definition~\ref{def:2}. The algorithm in Section~\ref{sec:alg1} is a greedy algorithm that outputs a single multiset compatible with $M$ with running time $O(nh)$. The algorithm in Section~\ref{sec:alg2} is able to output all compatible multisets up to reversal. Its running time is in general exponential as it relies on a breadth-first search to find all possible CWFs and solving a number of integer partition problems.
	
	\subsection{An algorithm that outputs a multiset of strings compatible with $M$}\label{sec:alg1} 
	To construct a multiset compatible with $M$, it suffices to find a CWF $f$ that is a solution to $M$. In Algorithm~\ref{alg:one} we construct such a CWF by assigning larger $w$ to $f_{2k-1}(l),k\in [h]$ in a greedy way as $l$ goes from $n$ to $0$.
	
	\begin{algorithm}
		\caption{Algorithm for getting one multiset of strings compatible with $M$}\label{alg:one}
		\begin{algorithmic}[1]
			\renewcommand{\algorithmicrequire}{\textbf{Input:}}
			\renewcommand{\algorithmicensure}{\textbf{Output:}}
			\Require A prefix-suffix compositions multiset $M$, the number of strings $h$, the length of the strings $n$, and the weight of the strings $\bar{w}$. 
			\Ensure  A multiset of $h$ strings of length $n$ and weight $\bar{w}$.
			\LComment{Initialization:}
			\State $a_{0,0}\gets2h$
			\State For $(l,w)\in \llb n\rrb^2\setminus\{0,0\}$, let $a_{l,w}$ be the number of pairs $(l-w,w)$ in {$M$}.
			\LComment{Construct a CWF $f$:}
			\State $f\gets \emptyset$
			\ForAll{$m\in[2h]$}
			\State $f(0,m)\gets 0$ \label{alg:line:f0}
			\State $f(n,m)\gets \bar{w}$ \label{alg:line:fm}
			\EndFor
			\For{$l=n-1$ to $1$}
			\For{$w=l$ to $1$}
			\State $s_{l,w+1} \gets \sum_{v=w+1}^{l} a_{l, v}$ 
			\ForAll{$k\in [1+s_{l,w+1}, \min\{a_{l,w}+s_{l,w+1},h\}]$}\label{alg:line:k}
			\State $f(l,2k-1)=w$ \label{alg:line:f}
			\State $f(n-l,2k)=\bar{w}-w$ \label{alg:line:co}
			\EndFor
			\EndFor
			\EndFor
			\LComment{Construct the multiset of strings corresponding to $f$:}
			\State $H_f \gets \{\bm{t}_j={t}_{j,1}\ldots{t}_{j,n} \mid t_{j,l} = f(l,2j-1)-f(l-1,2j-1)\text{ for all } l \in [n], j \in [h]\}$
			\State \Return ${H}_f$ 
		\end{algorithmic} 
	\end{algorithm}
	By Remark~\ref{re:M-f}, Algorithm~\ref{alg:one} produces a multiset of strings compatible with $M$ if the function $f$ constructed in the algorithm is a CWF and $|A_f(l,w)|=a_{l,w}$ for $(l,w)\in \llb n\rrb$.
	
	\begin{claim}\label{cl:one-cwf}
		The function $f$ constructed in Algorithm~\ref{alg:one} is a CWF.
	\end{claim}
	\begin{proof}
		Let us first show that $f$ is a mapping from $T$ to $\llb n\rrb$.
		Noticing Line~\ref{alg:line:f0} and \ref{alg:line:fm} in the algorithm, it suffices to show that for each $l\in [n-1],k\in [h]$, there exists $w\in[l]$ such that $1+s_{l,w+1} \leq k \leq a_{l,w}+s_{l,w+1}$. Note that $s_{l,w+1}$ is a non-increasing function of $w$. So as $w$ decreases from $l$ to $1$, $s_{l,w+1}$ increases from $0$ to at most $2h$. Moreover, $a_{l,w}+s_{l,w+1}=s_{l,w}$. Therefore, Line~\ref{alg:line:f} and \ref{alg:line:co} are well defined for each $l\in [n-1],k\in [h]$. It remains to show $f$ as constructed satisfies the conditions required in Definition~\ref{def:1}. Clearly, Item~\ref{ite:def:1:1} in the definition is satisfied according to Line~\ref{alg:line:f0} and Item~\ref{ite:def:1:3} is satisfied by Line~\ref{alg:line:f0}, \ref{alg:line:fm} and \ref{alg:line:co}.  
		
		As for Item~\ref{ite:def:1:2} in Definition~\ref{def:1}, it suffices to show that if $f(l,2k-1)=w$ then $f(l-1,2k-1)\in\{w,w-1\}$. According to Line~\ref{alg:line:k}, it suffices to show if $k\in [1+s_{l,w+1}, \min\{s_{l,w},h\}]$ then $k\in [1+s_{l-1,w+1}, \min\{s_{l-1,w},h\}]\cup [1+s_{l-1,w}, \min\{s_{l-1,w-1},h\}]=[1+s_{l-1,w+1},\min\{s_{l-1,w-1},h\}]$.
		From \eqref{eq:vertical}, we have 
		\begin{align*}
			s_{l,w+1}=\sum_{v=w+1}^{l} a_{l,v} \geq \sum_{v=w+1}^{l-1} a_{l-1,v}=s_{l-1,w+1}.
		\end{align*}
		At the same time, from \eqref{eq:horizontal}, we have 
		\begin{align*}
			s_{l-1,w-1}=\sum_{v=w-1}^{l-1} a_{l-1,v}\geq \sum_{v=w}^{l} a_{l,v}\geq s_{l,w}.
		\end{align*}
		Therefore, $[1+s_{l,w+1}, \min\{s_{l,w},h\}] \subset [1+s_{l-1,w+1},\min\{s_{l-1,w-1},h\}]$.
	\end{proof}

	Next, we would like to show that $|A_f(l,w)|=a_{l,w}$. Before that, let us make some simple observations.
	Since $M$ is the prefix-suffix compositions of $h$ strings of length $n$, for each $l\in [n]$ the number of prefixes and suffixes of length $l$ with weights in $\llb l\rrb$ is equal to $2h$. Moreover, since the $h$ strings are of the same weight $\bar{w}$, for each $l\in [n]$ and $w\in \llb \bar{w}\rrb$, the number of prefixes and suffixes of length $l$ with weight $w$ is the same as the number of suffixes and prefixes of length $n-l$ with weight $\bar{w}-w$. These observations extend to the case where $l=0$ since $a_{0,0}=2h$ by Definition~\ref{def:num-a}. Thus, we have the following proposition.
	
	\begin{prop}\label{prop:p2}
		\begin{enumerate}[label=(\roman*)]
			\item $\sum_{w=0}^l a_{l,w} = 2h$ for all $l \in \llb n \rrb$. \label{prop:p2:ite:1}
			\item If $M$ is the prefix-suffix composition of strings with constant weight, then $a_{l,w}=a_{n-l, \bar w - w}$ for all $l \in \llb n \rrb$ and $w\in \llb \bar{w}\rrb$. \label{prop:p2:ite:co}
		\end{enumerate}
	\end{prop}

	\begin{claim}\label{cl:one-A}
		For $(l,w)\in \llb n\rrb^2$, it holds that $|A_f(l,w)|=a_{l,w}$.
	\end{claim}
	\begin{proof}
		From Line~\ref{alg:line:f0} and \ref{alg:line:fm} in Algorithm~\ref{alg:one}, we have $|A_f(0,0)|=|A_f(n,\bar{w})|=2h=a_{0,0}=a_{n,\bar{w}}$. 
		In what follows, let $l\in[n-1]$. It is clear that $|A_f(l,w)|=0=a_{l,w}$ for all $w> l$.
		Since $s_{l,w+1}$ increases from $0$ to at most $2h$ as $w$ goes from $l$ to $1$, there exists $w_*\in[l]$ such that $s_{l,w_*+1}<h$ and $s_{l,w_*}\geq h$. From Line~\ref{alg:line:k} and \ref{alg:line:f} in the algorithm, for each $w \in [w_*+1, l]$ we have 
		\begin{align*}
			A_f(l,w)=\{2k-1\mid 1+s_{l,w+1}\leq k \leq a_{l,w}+s_{l,w+1} \}.
		\end{align*}
		Therefore, for all $l\in[n-1]$ and $w\in [w_*+1,l]$ we have $|A_f(l,w)|=a_{l,w}$.
		
		Consider the case where $w\in [0,w_*-1]$. Since $s_{l,w+1}\geq s_{l,w_*}\geq h$, we have 
		\begin{align}
			h&\geq 2h-s_{l,w+1}\nonumber\\
			&= 2h-\sum_{v=w+1}^{\bar{w}}a_{l,v}\nonumber\\
			&= 2h-\sum_{v=w+1}^{\bar{w}}a_{n-l,\bar{w}-v}\label{eq:co}\\
			&= 2h-\sum_{v=0}^{\bar{w}-w-1}a_{n-l,v}\nonumber\\
			&= \sum_{v=0}^{n-l}a_{n-l,v}-\sum_{v=0}^{\bar{w}-w-1}a_{n-l,v}\label{eq:sum}\\
			&= \sum_{v=\bar{w}-w}^{n-l}a_{n-l,v}\nonumber\\
			&= s_{n-l,\bar{w}-w},\nonumber
		\end{align} where \eqref{eq:co} follows by \ref{prop:p2:ite:co} in  Proposition~\ref{prop:p2} and \eqref{eq:sum} follows by Proposition~\ref{prop:p2:ite:1} in Proposition~\ref{prop:p2}.
		From Line~\ref{alg:line:k} and \ref{alg:line:co}, for each $w \in [0,w_*-1]$ we have 
		\begin{align*}
			A_f(l,w)=\{2k\mid 1+s_{n-l,\bar{w}-w+1}\leq k \leq a_{n-l,\bar{w}-w}+s_{n-l,\bar{w}-w+1} \}.
		\end{align*}
		Therefore, for all $l\in[n-1]$ and $w\in [0,w_*-1]$ we have $|A_f(l,w)|=a_{n-l,\bar{w}-w}=a_{l,w}$
		
		Lastly, note that $a_{l,w_*}+s_{l,w_*+1}\geq h$, and similarly to the above calculations that lead to $h\geq s_{n-l,\bar{w}-w}$ for $w\in[0,w_*-1]$, one can also obtain $h<2h-s_{l,w_*+1}=a_{n-l,\bar{w}-w_*}+s_{n-l,\bar{w}-w_*+1}$. Then from Line~\ref{alg:line:k}, \ref{alg:line:f}, and \ref{alg:line:co} we have
		\begin{align*}
			A_f(l,w_*)=\{2k-1\mid 1+s_{l,w_*+1}\leq k \leq h \}\cup \{2k\mid 1+s_{n-l,\bar{w}-w_*+1}\leq k \leq h \}.
		\end{align*}
		Therefore, for all $l\in[n-1]$ we have 
		\begin{align}
			|A_f(l,w_*)|&=h-s_{l,w_*+1}+h-s_{n-l,\bar{w}-w_*+1}\nonumber\\
			&= 2h -\sum_{v=w_*+1}^{\bar{w}}a_{l,v}-\sum_{v=\bar{w}-w_*+1}^{\bar{w}}a_{n-l,v}\nonumber\\
			&= 2h -\sum_{v=w_*+1}^{\bar{w}}a_{l,v}-\sum_{v=\bar{w}-w_*+1}^{\bar{w}}a_{l,\bar{w}-v}\label{eq:cos}\\
			&= 2h -\sum_{v=w_*+1}^{\bar{w}}a_{l,v}-\sum_{v=0}^{w_*-1}a_{l,v}\nonumber\\
			&= a_{l,w_*},\label{eq:sums}
		\end{align} where \eqref{eq:cos} follows by \ref{prop:p2:ite:co} in  Proposition~\ref{prop:p2} and \eqref{eq:sums} follows by Proposition~\ref{prop:p2:ite:1} in Proposition~\ref{prop:p2}.
		Hence, for all $l\in[n-1]$ and $w\in[l]$, we have $|A_f(l,w)|=a_{l,w}$.
	\end{proof}
	
	As a consequence of Claim~\ref{cl:one-cwf} and $\ref{cl:one-A}$, we have the following theorem.
	\begin{theorem}
		The output of Algorithm~\ref{alg:one} is a multiset of strings compatible with $M$.
	\end{theorem}
	
	Algorithm~\ref{alg:one} is an efficient algorithm with time complexity $O(nh)$, although it can only produce one multiset compatible with $M$. So it may not be helpful if one desires all compatible multisets. Nevertheless, let us mention one important application of Algorithm~\ref{alg:one}. In Theorem~\ref{UniqueThm}, the necessary and sufficient conditions for unique reconstruction given the prefix-suffix compositions $M$ are described in terms of a CWF rather than $M$ itself. Therefore, to determine the unique reconstructibility of $M$ using Theorem~\ref{UniqueThm}, it is necessary that one should be able to come up with a CWF solution to $M$. Algorithm~\ref{alg:one} does exactly what is needed for this purpose.
	
	Moreover, when one has a CWF $f$ solution to $M$ at hand, in view of Lemma~\ref{le:same-str} and \ref{notUnique}, it is tempting to use the swap operation as defined in Definition~\ref{def:swap} to enumerate all possible compatible multisets up to reversal. However, it is in general not easy to keep track of the swap operations. In the next subsection, we take a different route to constructing all compatible multisets by utilizing the inherent symmetry of the constant-length constant-weight strings, bypassing the difficulty brought about by the complexity of swap operations.

	\subsection{An algorithm that outputs all multisets of strings compatible with $M$}\label{sec:alg2}
	As mentioned before, to find a multiset of strings compatible with $M$, one may plot the elements of the multiset $M$ on the two dimensional grid $T$ and construct a CWF $f$ such that it passes each point $(l,w)$ exactly $a_{l,w}$ times on the grid. Moreover, one may infer the behavior of the component functions $\{f_m\}$ from the numbers $\{b_{l,w}\},\{c_{l,w}\}$. Therefore, to obtain all possible multisets of strings (up to reversal) that are compatible with $M$, one may examine all possible behaviors of $\{f_{m}\}$ based on $\{a_{l,w}\},\{b_{l,w}\},\{c_{l,w}\}$.
	
	Since all the $h$ strings that give rise to $M$ have the same length $n$ and the same weight $\bar{w}$, the graph of the component function $f_m$ is the same as that of $f_{m^*}$ when $f_m$ is rotated $180$ degrees around $(n/2,\bar{w}/2)$. As a result of this rotational symmetry, given the values of $f_m(l)$ for all $m\in[2h]$ and $l\in[0,\lfloor n/2\rfloor]$, the remaining values of $f_{m}(l),l\in[\lfloor n/2 \rfloor+1, n]$ can be fully determined for all $m\in[2h]$. Thus, it suffices to reconstruct all possible $\{f_m\}$ from the mid point $n/2$ to $0$, and then extend them from $n/2$ to $n$. However, there is one catch. The reason why such extension is possible is that $f_{m},f_{m^*}$ capture the running weight starting from the two ends of a single string. But for functions $g_{i},g_{j}\colon \llb\lfloor n/2\rfloor\rrb\to \llb \bar{w}\rrb,i\neq j$ reconstructed from $M$, it is in general not clear whether $g_{i},g_{j}$ capture the weight information of the same string. Nevertheless, by the rotational symmetry, $g_{i}$ and $g_j$ capture the weight information of the same string only if their median weights sum to $\bar{w}$ when they are extended. Therefore, one may identify $h$ pairs of functions from the $2h$ component functions $\{g_i\}$ reconstructed from $M$ such that the sum of median weights within each pair is $\bar{w}$. With the identification of such pairs, the resulting CWF formed by $\{g_i\}$ corresponds to a multiset of strings compatible with $M$. Thus, to obtain all compatible multisets, one needs to enumerate all possible ways of forming pairs that satisfy the median weight constraint.
	
	Based on the above discussion, our algorithm of constructing all multisets of strings compatible with $M$ are divided into two stages. In the first stage, which we call the scan stage, all possible ``half strings'' are generated based on $M$. In the second stage, which we call the assembly stage, pairs of ``half strings'' are combined to form ``full strings''. The details of the two stages are described below. For ease of discussion, below the subscript of the component functions will referred to as the label.
	
	\subsubsection*{Scan stage}	
	In the scan stage, we keep track of the behaviors of the component functions from the mid point $n/2$ to $0$. Consider the case where $n$ is even. For each $w\in\llb \bar{w}\rrb$, $a_{{n/2},w}$ indicates the number of component functions that evaluate to $w$ at $n/2$. Moreover, we have $\sum_{w=0}^{\bar{w}}a_{{n/2},w}=2h$. As there are $2h$ component functions, we may partition the $2h$ labels into disjoint subsets of sizes $a_{n/2,w}, w\in\llb \bar{w}\rrb$. If $n$ is odd, $a_{n/2,w}$ is undefined, but the behavior of the component functions at $n/2$ can be determined by $b_{\lceil n/2\rceil,w},c_{\lceil n/2 \rceil,w}$. Since $\sum_{w=0}^{\bar{w}}(b_{\lceil n/2\rceil,w}+c_{\lceil n/2 \rceil,w})=2h$, we can partition the $2h$ labels into disjoint subsets of sizes $b_{\lceil n/2\rceil,w},c_{\lceil n/2 \rceil,w}, w\in\llb \bar{w}\rrb$. More precisely, as the first step of the scan stage, we
	construct a collection $\{P(t/2)\mid t=0,\ldots, 2\bar{w}\}$ of disjoint subsets of $[2h]$ such that $\bigcup_{t=0}^{2\bar{w}}P(t/2)=[2h]$ and that if $n$ is even,
	\begin{equation}
		\begin{cases}
			|P(t/2)| = a_{n/2,t/2}, & \text{$t$ is even},\\
			|P(t/2)| = 0, & \text{$t$ is odd};
		\end{cases}\label{eq:scan-init-even}
	\end{equation}	
	if $n$ is odd,
	\begin{equation}
		\begin{cases}
			|P(t/2)| = b_{\lceil n/2\rceil,t/2}, &  \text{$t$ is even},\\
			|P(t/2)| = c_{\lceil n/2\rceil,\lceil t/2\rceil}, &  \text{$t$ is odd}.
		\end{cases}\label{eq:scan-init-odd}
	\end{equation}
	
	Observe that the elements in the set $P(t/2)$ are the labels of the component functions whose median weight equals $t/2$. So we basically reconstruct the values of $2h$ component functions at $n/2$ by constructing the collection $\{P(t/2)\}$. Given the values of the component functions at $n/2$, we reconstruct their values at $l\leq n/2$ according to $\{b_{l,w}\},\{c_{l,w}\}$ as $l$ goes from $\lfloor n/2\rfloor$ to $0$. Specifically, we keep track of the labels of the component functions as we assign values to the component functions at $l=\lfloor n/2\rfloor, \ldots, 0$ according to $\{b_{l,w}\},\{c_{l,w}\}$, and obtain finer partitions of the $2h$ labels as $l$ goes to $0$. The bookkeeping of the partitions is done by a function $F$ that maps each $(l,w)\in \llb \lfloor n/2\rfloor\rrb \times \llb \bar{w}\rrb $ to a collection of disjoint subsets of the $2h$ labels. The labels in the these disjoint subsets correspond to component functions that evaluate to $w$ at $l$. Moreover, the subsets in $F(l,w),w\in\llb \bar{w}\rrb$ are all disjoint, and we have $\bigcup_{w=0}^{\bar{w}}\bigcup_{J\in F(l,w)}J=[2h]$ for  $l\leq n/2$. The construction of $F$ is described below. 
	
	By construction of $\{P(t/2)\}$, the component functions that evaluate to $\bar{w}$ at $\lfloor n/2\rfloor$ are those with labels in $P(\bar{w})$, and for each $w\in\llb \bar{w}-1\rrb$ the component functions that evaluate to $w$ at $\lfloor n/2\rfloor$ are those with labels in $P(w)$ and $P(w+1/2)$. Therefore, 
	\begin{align*}
		F(\lfloor n/2\rfloor,\bar{w})&=\{P(\bar{w})\},\\
		F(\lfloor n/2\rfloor,w)&=\{P(w),P(w+1/2)\},\quad w\in\llb\bar{w}-1\rrb.
	\end{align*}
	As the value of a component function at $l-1$ may remain the same as or decrease by one from the value at $l$, given $F(l,w),w\in\llb\bar{w}\rrb$, we can further partition each subset in $F(l,w),w\in\llb\bar{w}\rrb$ into two subsets of sizes $b_{l,w},c_{l,w}$ for $l=\lfloor n/2\rfloor,\ldots, 0$. Eventually, we obtain the set $F(0,0)$ in which every element is a subset of the $2h$ labels for which the corresponding component functions have exactly the same values at $l=0,\ldots,\lfloor n/2\rfloor$. Moreover, component functions with labels in different element in $F(0,0)$ are not equal. At this point, the behaviors of the $2h$ component functions are determined over $\llb\lfloor n/2\rfloor\rrb$. In particular, one can define $2h$ component functions $\{g_m\}$ over $\llb\lfloor n/2\rfloor\rrb$ to be 
	\[
	g_{m}(l)= w \text{ if } m \in {\bigcup_{J \in F(l, w)} J}.
	\] 	
	
	Note that there are different ways of partitioning subsets in $F(l,w),w\in\llb\bar{w}\rrb$ and each of them leads to a distinct $F$. However, we are only interested in those $F$'s that result in distinct ``half strings'', i.e., distinct multiset $\{g_m\}$. In other words, we only care about the numbers of labels for which the corresponding component functions are the same over $\llb \lfloor n/2\rfloor\rrb$. In fact, this is the reason why we only stipulate the size of the subsets in the initial partition $\{P(t/2)\}$. In order to construct all possible $F$, each of which leads to a distinct multiset $\{g_m\}$, we need to enumerate different ways of partitioning subsets in $F(l,w),w\in\llb\bar{w}\rrb$. 
	This is accomplished as follows. Let $q=|F(l,w)|$ and write $F(l,w)=\{J_1,\ldots,J_q\}$. Further, let $K_i\subset J_i$ be the labels for which the corresponding component functions have value equal to $w-1$ at $l-1$. Denote $|K_i|$ by $x_i$. Since $c_{l,w}$ is the number of component functions that have values equal to $w-1$ at $l-1$ and have values equal to $w$ at $l$, we have 
	\begin{align}
		\sum_{i=1}^{q}x_i=c_{l,w}.\label{eq:partition}
	\end{align}
	Every solution to \eqref{eq:partition} such that $x_i\in \llb |P_i| \rrb$ gives rise to a distinct partition of the subsets in $F(l,w)$. By enumerating all possible solutions to \eqref{eq:partition} for every $l \in [\lfloor n/2\rfloor]$ and $w\in\llb \bar{w}\rrb$, we are able to find the set $\mathcal{F}$ of all possible $F$ that leads to distinct $\{g_m\}$ via breadth-first search. The scan stage is formally stated in Algorithm~\ref{alg:all-1}. 
	
	As a consequence of the scan stage, we obtain a set of all possible ``half strings'' from $M$ in the sense of the following claim.
	
	\begin{claim}\label{cl:scan}
		Let $\{\bm{t}_{j}\mid j\in[h]\}$ be a multiset of strings compatible with $M$ and define the multiset of length-$\lfloor n/2\rfloor$ prefixes and suffixes of $\bm{t}_{j},j\in[h]$ to be $S=\{\bm{s}_{2j-1}=\bm{t}_j[\lfloor n/2\rfloor],\bm{s}_{2j}=\cev{\bm{t}_{j}}[\lfloor n/2\rfloor]\mid j\in[h] \}$. Let $S'$ be the underlying set of $S$, i.e., $S'$ is the set of distinct strings in $S$.
		Then there exists $F\in\mathcal{F}$ output by Algorithm~\ref{alg:all-1} such that there is a bijection between $F(0,0)$ and $S'$ that maps $J\in F(0,0)$ to $\bm{s}$ with $|J|=|\{m\mid \bm{s}_m=\bm{s},\bm{s}_m\in S,m\in[2h]\}|$. 
		
		In other words, there exists $F\in\mathcal{F}$ such that every element $J$ in $F(0,0)$ can be identified with a distinct string $\bm{s}$ in $S$ whose multiplicity in $S$ equals $|J|$.
	\end{claim}
	\begin{proof}
		Let $f$ be the CWF induced by $\{\bm{t}_{j}\mid j\in[h]\}$ with $f_{2j-1}$ being induced by the running weight of $\bm{t}_{j}$ and $f_{2j}$ by the running weight of $\cev{\bm{t}_{j}}$. Then $f$ is a solution to $M$.
		Let us construct a function $\tilde{F}$ that maps each $(l,w)\in \llb \lfloor n/2\rfloor\rrb \times \llb \bar{w}\rrb $ to a collection of disjoint subsets of $[2h]$ dependent on $f$.
		Given $M$, we can construct a collection $\{P(t/2)\mid t\in \llb 2\bar{w}\rrb \}$ of disjoint subsets of $[2h]$ such that $\bigcup_{t=0}^{2\bar{w}}P(t/2)=[2h]$ and that satisfies \eqref{eq:scan-init-even} or \eqref{eq:scan-init-odd} based on the parity of $n$. Furthermore, there exists a permutation on $[2h]$ such that $P(t/2)$ is formed by $m\in[2h]$ for which $\med(f_m)=t/2$. Define 
		\begin{align*}
			\tilde{F}(\lfloor n/2\rfloor,\bar{w})&=\{P(\bar{w})\},\\
			\tilde{F}(\lfloor n/2\rfloor,w)&=\{P(w),P(w+1/2)\},\quad w\in\llb\bar{w}-1\rrb.
		\end{align*}
		For $l\in[ \lfloor n/2\rfloor ]$, define
		\begin{align*}
			\tilde{F}(l-1,\bar{w})&=\{A_f(l-1,\bar{w})\cap J\mid J\in\tilde{F}(l,\bar{w})\},\\
			\tilde{F}(l-1,w)&=\{A_f(l-1,w)\cap J\mid J\in\tilde{F}(l,w)\cup \tilde{F}(l,w+1)\},\quad w\in\llb\bar{w}-1\rrb,
		\end{align*} where $A_f(l,w)$ is as given in Definition~\ref{def:5}. It follows that $\bigcup_{\tilde{J}\in\tilde{F}(l,w)}\tilde{J}=A_f(l,w)$ for $l\in\llb \lfloor n/2\rfloor\rrb, w\in\llb \bar{w}\rrb$. In addition, $m,m'\in [2h]$ are in the same set $\tilde{J}\in \tilde{F}(l,w)$ if and only if the component functions $f_m$ and $f_{m'}$ have the same graph over $[l,\lfloor n/2\rfloor ]$. Therefore, $|\tilde{F}(0,0)|$ equals the number of distinct graphs over $\llb \lfloor n/2\rfloor\rrb$ of $f_m,m\in[2h]$, i.e., $|\tilde{F}(0,0)|=|S'|$. Furthermore, there is a bijection between $\tilde{F}(0,0)$ and $S'$ that maps $\tilde{J}\in \tilde{F}(0,0)$ to $\bm{s}\in S'$ with $\tilde{J}=\{m\mid \bm{s}_m=\bm{s},\bm{s}_m\in S,m\in [2h]\}$. 
		
		The set $\mathcal{F}$ output by Algorithm~\ref{alg:all-1} is the set of bookkeeping functions $F$ that keep track of all admissible behaviors of the component functions given $M$. Moreover, every element in $F(0,0)$ is a subset of the $2h$ labels for which the corresponding component functions have the same graph over $\llb \lfloor n/2\rfloor\rrb$. The construction of $P(t/2)$ in Line~\ref{alg:l:P} and $K_i$ in Line~\ref{alg:l:K} in Algorithm~\ref{alg:all-1} is oblivious of which labels in $[2h]$ to choose but dependent on the admissible sizes of the sets. Since the size of $P(t/2)$ must satisfy \eqref{eq:scan-init-even}, \eqref{eq:scan-init-odd} and the set $X$ constructed on Line~\ref{alg:l:X} enumerates all admissible sizes for $K_i$, there exists $F\in\mathcal{F}$ such that $|F(0,0)|=|\tilde{F}(0,0)|$ and a bijection between $F(0,0)$ and $\tilde{F}(0,0)$ that maps $J\in F(0,0)$ to $\tilde{J}\in \tilde{F}(0,0)$ with $|J|=|\tilde{J}|$.
	\end{proof}
	
	\begin{algorithm}
		\caption{Scan stage}\label{alg:all-1}
		\begin{algorithmic}[1]
			\renewcommand{\algorithmicrequire}{\textbf{Input:}}
			\renewcommand{\algorithmicensure}{\textbf{Output:}}
			\Require A prefix-suffix compositions multiset $M$, the number of strings $h$, the length of the strings $n$, and the weight of the strings $\bar{w}$. 
			\Ensure  A set $\mathcal{F}$ that corresponds to all multisets of ``half strings'' and a collection $\{P(t/2)\mid t=0,\ldots, 2\bar{w}\}$ of disjoint subsets of $[2h]$.
			\LComment{Initialization:}
			\State $a_{0,0}\gets2h$
			\State For $(l,w)\in \llb n\rrb^2\setminus\{0,0\}$, let $a_{l,w}$ be the number of pairs $(l-w,w)$ in {$M$}.
			\ForAll{$(l,w)\in \llb n\rrb^2$}
			\State $b_{l,w} \gets \sum_{v=w}^{l-1} a_{l-1,v}-\sum_{v=w+1}^{l} a_{l,v}$ 
			\State  $c_{l, w} \gets \sum_{v=w}^{l} a_{l,v} - \sum_{v=w}^{l-1} a_{l-1,v}$
			\EndFor
			\State Construct a collection $\{P(t/2)\mid t=0,\ldots, 2\bar{w}\}$ of disjoint subsets of $[2h]$ such that $\bigcup_{t=0}^{2\bar{w}}P(t/2)=[2h]$ and that satisfies \eqref{eq:scan-init-even} or \eqref{eq:scan-init-odd} based on the parity of $n$. \label{alg:l:P}	
			\State $F\gets\emptyset$.
			\State $F(\lfloor n/2\rfloor,\bar{w})\gets \{P(\bar{w})\}$
			\ForAll{$w\in\llb\bar{w}-1\rrb$}
			\State $F(\lfloor n/2\rfloor,w)\gets \{P(w),P(w+1/2)\}$
			\EndFor
			\State $\mathcal{F}\gets \{F\}$
			\LComment{Find finer partitions:}
			\For{$l=\lfloor n/2\rfloor$ to $1$}\label{alg:l:bfs1}
			\For{$w=\bar{w}$ to $0$}
			\State $\mathcal{G}\gets\emptyset$
			\ForAll{$F\in\mathcal{F}$}
			\State $q\gets |F(l,w)|$
			\State ${X}\gets \{(x_1,\ldots,x_q)\mid \sum_{i=1}^{q}x_i=c_{l,w}\text{ where }x_i\in \llb |J_i| \rrb, J_i\in F(l,w), i\in [q]\}$\label{alg:l:X}
			\ForAll{$(x_1,\ldots,x_q)\in{X}$}
			\State Let $K_i$ be a subset of $J_i$ such that $|K_i|=x_i$ where $J_i\in F(l,w),i\in[q]$ and construct a function $G$ such that \label{alg:l:K}
			\begin{enumerate}
				\item $G(l',w')=F(l',w')$ for $(l',w')\in[l,\lfloor n/2\rfloor] \times \llb \bar{w}\rrb $;
				\item $G(l-1,w')=F(l-1,w')$ for $w'\in [w+1,\bar{w}]$;
				\item $G(l-1,w)=F(l-1,w)\cup\{J_i\setminus K_i\mid i\in [q]\}$;
				\item $G(l-1,w-1)=\{K_i\mid i\in [q]\}$ if $w\geq 1$.
			\end{enumerate}
			\State ${\mathcal{G}}\gets{\mathcal{G}}\cup\{G\}$
			\EndFor
			\EndFor
			\State $\mathcal{F}\gets\mathcal{G}$
			\EndFor
			\EndFor\label{alg:l:bfs2}
			\State Return $\mathcal{F}$ and $\{P(t/2)\mid t=0,\ldots, 2\bar{w}\}$
		\end{algorithmic} 
	\end{algorithm}
	
	\subsubsection*{Assembly stage} In the assembly stage, we construct CWFs for each $F\in \mathcal{F}$ by identifying pairs in $\{g_m\}$ whose median weights sum to $\bar{w}$. As mentioned in the scan stage, $F(0,0)$ is a partition of $[2h]$, and for each $J\in F(0,0)$, the component functions with labels in $J$ have the same graph over $\llb \lfloor n/2\rfloor\rrb$. As we would like to form pairs of component functions based on their median weights, it is helpful to group the elements of $F(0,0)$ based on the median weight. More precisely, for each possible median weight $w=0,1/2,1,\ldots,\bar{w}$, we construct a collection $R_w$ of sets for which the corresponding component functions have median weight $w$, given by 
	\begin{align*}
		R_{w}=\{J\mid J\in F(0,0),J\subset P(w) \}.
	\end{align*}
	Let $r_w=|R_w|$. Since different elements in $F(0,0)$ correspond to component functions with different graphs, $r_w$ is the number of distinct component functions that have median weight $w$. Moreover, each element $R_{w,i}\in R_w, i\in[r_w]$ is a set of labels for which the corresponding component functions have median weight $w$ and the same graph over $\llb \lfloor n/2 \rfloor\rrb$.
	
	By the rotational symmetry, two component functions capture the weight information of the same string only if their median weights sum to $\bar{w}$. Therefore, a label in $R_{w,i}\in R_{w}, i\in[r_w]$ must be paired with a label in $R_{\bar{w}-w,j}\in R_{\bar{w}-w}, j\in[r_{\bar{w}-w}]$ in order to combine two ``half strings'' into a single ``full string''.  Formally, the pairing of labels can described by a permutation $\sigma$ on $[2h]$ such that if $u\in R_{w,i}$ is paired with $v\in R_{\bar{w}-w,j}$ then $\sigma(u)=m,\sigma(v)=m^*$ for some $m\in[2h]$, i.e., $\sigma(u)^*=\sigma(v)$.

	To enumerate all possible ways of forming pairs that satisfy the median weight constraint, we need to consider different ways of pairing a component function of median weight $w\in \{0,1/2,1,\ldots,\bar{w}\}$ with a component function of median weight $\bar{w}-w$. 
	Let us first consider the case where $w\in\{0,1/2,1,\ldots,(\bar{w}-1)/2\}$.
	For $ i\in[r_w], j\in[r_{\bar{w}-w}]$, let $y_{w,i,j}$ be the number of labels chosen in $R_{w,i}$ to be paired with labels in $R_{\bar{w}-w,j}$. 
	Then $(y_{w,i,j})_{i\in[r_w],j\in[r_{\bar{w}-w}]}$ must satisfy
	\begin{align}
		&\sum_{i=1}^{r_{w}}y_{w,i,j}=|R_{\bar{w}-w,j}|,\quad j\in[r_{\bar{w}-w}],\label{eq:assembly-y1}\\
		&\sum_{j=1}^{r_{\bar{w}-w}}y_{w,i,j}=|R_{w,i}|,\quad i\in[r_w].\label{eq:assembly-y2}
	\end{align}
	For each solution $(y_{w,i,j})_{i\in[r_w],j\in[r_{\bar{w}-w}]}$ to \eqref{eq:assembly-y1} and \eqref{eq:assembly-y2}, we partition $R_{\bar{w}-w,j}$ into disjoint subsets $\{V_{w,i,j}\mid i\in[r_{w}]\}$ and $R_{w,i}$ into disjoint subsets $\{U_{w,i,j}\mid j\in[r_{\bar{w}-w}]\}$ such that $|V_{w,i,j}|=|U_{w,i,j}|=y_{w,i,j}$. The labels in $V_{w,i,j}$ are then paired with the labels in $U_{w,i,j}$.
	
	Consider the case where $w=\bar{w}/2$. In this case, the labels in $R_{\bar{w}/2}$ need to be paired with each other so we have a slightly different integer partition problem. For $ i\in[r_{\bar{w}/2}], j\in[r_{\bar{w}/2}]$, let $y_{\bar{w}/2,i,j}$ be the number of labels chosen in $R_{\bar{w}/2,i}\in R_{\bar{w}/2}$ to be paired with labels in $R_{\bar{w}/2,j}\in R_{\bar{w}/2}$. 
	Then $y_{\bar{w}/2,i,i}$ must be even for all $i$ and $y_{\bar{w}/2,i,j}=y_{\bar{w}/2,j,i}$ for all $i\neq j$. Moreover, $(y_{\bar{w}/2,i,j})_{i\in[r_{\bar{w}/2}],j\in[r_{\bar{w}/2}]}$ must satisfy that
	\begin{align}
		\sum_{j=1}^{r_{\bar{w}/2}}y_{\bar{w}/2,i,j}=|R_{\bar{w}/2,i}|,\quad i\in[r_{\bar{w}/2}].\label{eq:assembly-y3}
	\end{align}
	For each solution $(y_{\bar{w}/2,i,j})_{i\in[r_{\bar{w}/2}],j\in[r_{\bar{w}/2}]}$ to \eqref{eq:assembly-y3}, we partition $R_{\bar{w}/2,j}$ into disjoint subsets $\{U_{\bar{w}/2,i,j}\mid j\in[r_{\bar{w}/2}]\}$ such that $|U_{i,j}|=y_{\bar{w}/2,i,j}$. The labels in $U_{\bar{w}/2,j,i}$ are then paired with the labels in $U_{\bar{w}/2,i,j}$ for $i\neq j$, and the labels in $U_{\bar{w}/2,i,i}$ are organized into $y_{\bar{w}/2,i,i}/2$ pairs arbitrarily.
	
	Let $Y_w=\{(y_{w,i,j})_{i\in[r_w],j\in[r_{\bar{w}-w}]}\}$ be the set of all solutions to the integer partition problem associated with $w\in\{0,1/2,1,\ldots,\bar{w}/2\}$ and let $Y=Y_0\times Y_{1/2}\times Y_1\times \cdots Y_{\bar{w}/2}$. Then each $(y_{w,i,j})\in Y$ corresponds to a distinct way of forming pairs of the component functions such that the median weight constraint is satisfied. Specifically, since $R_{t/2,i},t\in\llb 2\bar{w}\rrb,i\in[r_{t/2}]$ are disjoint and 
	\[
	\bigcup_{t=0}^{2\bar{w}} \bigcup_{i=1}^{r_{t/2}} R_{t/2,i}=\bigcup_{J\in F(0,0)}J=[2h],
	\]
	one can easily define a permutation $\sigma$ on $[2h]$ such that if $u\in R_{w,i}$ is paired with $v\in R_{\bar{w}-w,j}$ then $\sigma(u)=m,\sigma(v)=m^*$ for some $m\in[2h]$. Furthermore, given $\sigma$, a CWF $f$ can be determined by combining the paired component functions, i.e., those with labels $u,v\in[2h]$ satisfying $\sigma(u)^*=\sigma(v)$. The corresponding multiset $H_f$ can then be found using Definition~\ref{def:2}. The details are presented in Algorithm~\ref{alg:all-2}.
	
	\begin{theorem}
		The output $\mathcal{H}$ of running Algorithm~\ref{alg:all-1} followed by Algorithm~\ref{alg:all-2} is the set of all multisets compatible with $M$ up to reversal.
	\end{theorem}
	\begin{proof}
		Let $\{\bm{t}_{j}\mid j\in[h]\}$ be a multiset of strings compatible with $M$ and define the multiset of length-$\lfloor n/2\rfloor$ prefixes and suffixes of $\bm{t}_{j},j\in[h]$ to be $S=\{\bm{s}_{2j-1}=\bm{t}_j[\lfloor n/2\rfloor],\bm{s}_{2j}=\cev{\bm{t}_{j}}[\lfloor n/2\rfloor]\mid j\in[h] \}$. Let $S'$ be the underlying set of $S$.
		By Claim~\ref{cl:scan}, there exists $F\in\mathcal{F}$ output by Algorithm~\ref{alg:all-1} such that there is a bijection $\pi$ between $F(0,0)$ and $S'$ that maps $J\in F(0,0)$ to $\bm{s}\in S'$ with $|J|=|\{m\mid \bm{s}_m=\bm{s},\bm{s}_m \in S,m\in [2h]\}|$. Denote the set $J$ mapped to $\bm{s}$ under $\pi$ by $J_{\bm{s}}$ and denote $\{m\mid \bm{s}_m=\bm{s},\bm{s}_m\in S,m\in [2h]\}$ by $I_{\bm{s}}$. Since $[2h] = \bigcup_{\bm{s}\in S'}J_{\bm{s}} = \bigcup_{\bm{s}\in S'}I_{\bm{s}}$, a permutation $\tilde{\sigma}$ on $[2h]$ can be further constructed such that it is a bijection between $J_{\bm{s}}$ and $I_{\bm{s}}$ for every $\bm{s}\in S'$.
		
		In Algorithm~\ref{alg:all-2}, given $F$, all possible permutations for pairing labels in $R_w$ and $R_{\bar{w}-w}$ for all $w\in\{0,1/2,1,\ldots,\bar{w}/2\}$ are found. In particular, there exists a permutation $\sigma$ such that for any $u,v\in [2h]$ satisfying $\tilde{\sigma}(u)^*=\tilde{\sigma}(v)$, it holds that $\sigma(u)^*=\sigma(v)$. On Line~\ref{alg:l:f1} to \ref{alg:l:f2} in Algorithm~\ref{alg:all-2}, a function $f$ is constructed according to $F,\sigma$, and it is easy to verify $f$ is a CWF. Thus, the multiset corresponding to $f$ constructed on Line~\ref{alg:l:Hf} in Algorithm~\ref{alg:all-2} satisfies $H_f=\{\bm{t}_{j}\mid j\in[h]\}$. It follows that any multiset compatible with $M$ is in the output $\mathcal{H}$.
		
		It remains to check the elements in $\mathcal{H}$ are all distinct. In fact, let us show the CWFs constructed in Algorithm~\ref{alg:all-2} as multisets $\{f_m\}$ are distinct. Let $F_1,F_2\in\mathcal{F}$ with $F_1\neq F_2$. Then $F_1,F_2$ correspond to distinct sets of ``half strings'', and any pairing permutations $\sigma_1,\sigma_2$ admissible for $F_1,F_2$, respectively, lead to distinct multisets of component functions. Furthermore, if $\sigma_1,\sigma_2$ are two different pairing permutations admissible for $F_1$, then the two multiset of component functions resulted from $\sigma_1,\sigma_2$ are also different since each element $R_{w,i}\in R_{w}$ corresponds to a distinct ``half string''. Therefore, all CWFs constructed in Algorithm~\ref{alg:all-2} are distinct as multisets. Moreover, since a multiset and its reversals induce the same multiset of component functions, if a multiset is in $\mathcal{H}$ then any of its reversals is not in $\mathcal{H}$.
		Hence, $\mathcal{H}$ is a set of all multisets compatible with $M$ up to reversal.
	\end{proof}
	\begin{algorithm}
		\caption{Assembly stage}\label{alg:all-2}
		\begin{algorithmic}[1]
			\renewcommand{\algorithmicrequire}{\textbf{Input:}}
			\renewcommand{\algorithmicensure}{\textbf{Output:}}
			\Require A set $\mathcal{F}$ that corresponds to all multisets of ``half strings'' and a collection $\{P(t/2)\mid t=0,\ldots, 2\bar{w}\}$ of disjoint subsets of $[2h]$ from Algorithm~\ref{alg:all-1}.
			\Ensure  A collection $\mathcal{H}$ of multisets of $h$ strings of length $n$ and weight $\bar{w}$.
			\State $\mathcal{H}\gets \emptyset$
			\ForAll{$F\in\mathcal{F}$}
			\LComment{Group labels by the median weight $w$:}
			\ForAll{$w\in\{0,1/2,1,\ldots,\bar{w}\}$}
			\State $R_{w}\gets \{J\mid J\in F(0,0),J\subset P(w) \}$
			\State $r_w\gets |R_w|$
			\EndFor
			\LComment{Enumerate all possible ways of forming pairs by constructing sets $\{Y_w\}$:}
			\ForAll{$w\in\{0,1/2,1,\ldots,\bar{w}/2\}$}
			\If{$w\neq \bar{w}/2$}
			\State Find the set $Y_w$ formed of $(y_{i,j})_{i\in[r_w],j\in[r_{\bar{w}-w}]}$ such that $y_{i,j}\in\mathbb{N}$ and that \eqref{eq:assembly-y1} and \eqref{eq:assembly-y2} are satisfied.
			\Else
			\State Find the set $Y_{\bar{w}/2}$ formed of $(y_{i,j})_{i,j\in[r_{\bar{w}/2}]}$ such that $y_{i,i}\in 2\mathbb{N}$, $y_{i,j}=y_{j,i}\in\mathbb{N}$ for $i\neq j$ and that \eqref{eq:assembly-y3} is satisfied.
			\EndIf
			\EndFor
			\ForAll{$(y_{w,i,j})\in{Y_0\times Y_{1/2}\times Y_{1}\times\cdots\times Y_{\bar{w}/2}}$}
			\LComment{Construct a permutation $\sigma$ on $[2h]$ for each way of forming pairs:}
			\State $\sigma \gets \emptyset$
			\State $m\gets 1$
			\ForAll{$w\in\{0,1/2,1,\ldots,\bar{w}/2\}$}
			\If{$w\neq \bar{w}/2$}
			\State For each $i\in[r_w]$, partition $R_{w,i}\in R_{w}$ into subsets $\{U_{w,i,j}\mid j\in[r_{\bar{w}-w}]\}$ where $|U_{w,i,j}|=y_{w,i,j}$.  
			\State For each $j\in [r_{\bar{w}-w}]$, partition $R_{\bar{w}-w,j}\in R_{\bar{w}-w}$ into subsets $\{V_{w,i,j}\mid i\in[r_{w}] \}$ where $|V_{w,i,j}|=y_{w,i,j}$.
			
			\ForAll{$i\in[r_w],j\in[r_{\bar{w}-w}]$}
			\LComment{Construct a bijection between $U_{w,i,j}\cup V_{w,i,j}$ and $\{m,m+1\ldots, m + 2y_{w,i,j}-1\}$.}
			\State $\sigma(U_{w,i,j})\gets \{m, m+2,\ldots, m + 2y_{w,i,j}-2\}$
			\State $\sigma(V_{w,i,j})\gets \{m+1, m+3,\ldots, m + 2y_{w,i,j}-1\}$
			\State $m\gets m+2y_{w,i,j}$
			\EndFor
			\Else
			\State For each $i\in[r_{\bar{w}/2}]$, partition $R_{\bar{w}/2,i}\in R_{\bar{w}/2}$ into subsets $\{U_{w,i,j}\mid j\in [r_{\bar{w}/2}]\}$ where $|U_{w,i,j}|=y_{w,i,j}$.
			\ForAll{$i\in[r_{\bar{w}/2}],j\in[r_{\bar{w}/2}]$}
			\If{$i=j$}
			\State $\sigma(U_{\bar{w}/2,i,i})\gets \{m,m+1,\ldots,m+y_{\bar{w}/2,i,i}-1\}$
			\State $m\gets m+y_{\bar{w}/2,i,i}$
			\Else
			\State $\sigma(U_{\bar{w}/2,i,j})\gets \{m,m+2,\ldots,m+2y_{\bar{w}/2,i,j}-2\}$
			\State $\sigma(U_{\bar{w}/2,j,i})\gets \{m+1,m+3,\ldots,m+2y_{\bar{w}/2,i,j}-1\}$
			\State $m\gets m+2y_{\bar{w}/2,i,j}$
			\EndIf
			\EndFor
			\EndIf
			\EndFor
			\LComment{Construct a CWF $f$ using $F$ and $\sigma$:}
			\State $f\gets \emptyset$\label{alg:l:f1}
			\For{$l=0$ to $\lfloor n/2\rfloor$}
			\For{$w=0$ to $\bar{w}$}
			\For{$u\in \bigcup_{J\in F(l,w)}J$}
			\State $f(l,\sigma(u))=w$
			\EndFor
			\EndFor
			\EndFor
			\For{$l=\lfloor n/2\rfloor+1$ to $n$}
			\For{$m=1$ to $2h$}
			\State $f(l,m)=\bar{w}-f(n-l,m^*)$
			\EndFor
			\EndFor\label{alg:l:f2}
			\LComment{Construct the multiset of strings corresponding to $f$:}
			\State $H_f \gets \{\bm{t}_j={t}_{j,1}\ldots{t}_{j,n} \mid t_{j,l} = f(l,2j-1)-f(l-1,2j-1)\text{ for all } l \in [n], j \in [h]\}$\label{alg:l:Hf}
			\State $\mathcal{H}\gets\mathcal{H}\cup H_f$.
			\EndFor
			\EndFor
			\State \Return $\mathcal{H}$ 
		\end{algorithmic} 
	\end{algorithm}

	\section{Concluding remarks}\label{sec:conclusion}
	We propose to use cumulative weight functions to describe the prefix-suffix compositions of a multiset of binary strings, and facilitate this description to derive necessary and sufficient conditions for unique reconstruction of multisets of strings {of the same weight} up to reversal. Moreover, two reconstruction algorithms are presented. One is an efficient algorithm that outputs one multiset of strings compatible with the given prefix-suffix compositions and can be used to assist in determining the unique reconstructibility of the given compositions. The other one is able to output all admissible multisets up to reversal that are compatible with the given compositions. 
	
	Many problems on reconstruction of multiple strings remain open. For example, can one lift the constant-weight assumption and characterize the conditions for unique reconstruction of multiple strings from prefix-suffix compositions? In addition, if the prefix-suffix compositions are erroneous, can one design low-redundancy encoding schemes for the strings such that they can be recovered efficiently?
	
	\bibliographystyle{IEEEtran}
	\bibliography{MultiStringReconstructionFromAffixes}
	
	\appendices
	\section{Proof of Lemma~\ref{le:half-diff-0}}\label{app:half-diff-0}
		Recall that $m\in A(\bar{w}/2)$ if and only if $m^*\in A(\bar{w}/2)$ so the size of $A(\bar{w}/2)$ must be even. The case where $|A(\bar w/2)|=0$ is vacuously true and the case where $|A(\bar w/2)|=2$ follows from Proposition~\ref{prop:f-dual}. Below we assume $|A(\bar w/2)|\geq 4$.
		
		Let $\tilde{m}\in A(\bar{w}/2)$. If $f_{\tilde{m}}\neq f_{\tilde{m}^*}$, then by Proposition~\ref{prop:f-dual} there are exactly two maximal intervals between $f_{\tilde{m}}$ and $f_{\tilde{m}^*}$. Let $m\in A(\bar{w}/2)\setminus\{\tilde{m},\tilde{m}^*\}$. By the conditions in Theorem~\ref{UniqueThm}, there is at most one maximal interval between $f_{m}$ and $f_{\tilde{m}}$. We claim that there is exactly one maximal interval between them. Indeed, if $f_{m}=f_{\tilde{m}}$, then there are two maximal intervals between $f_{m}$ and $f_{\tilde{m}^*}$, contradicting the conditions in Theorem~\ref{UniqueThm}. 
		Similarly, one can show there is exactly one maximal interval between $f_{m}$ and $f_{\tilde{m}^*}$. 
		
		Suppose $\mathcal{G}(f_m,[\lfloor n/2 \rfloor])\neq \mathcal{G}(f_{\tilde{m}},[\lfloor n/2 \rfloor])$ and $\mathcal{G}(f_m,[\lfloor n/2 \rfloor])\neq \mathcal{G}(f_{\tilde{m}^*},[\lfloor n/2 \rfloor])$. Since there is exactly one maximal interval between $f_m,f_{\tilde{m}}$ and exactly one maximal between $f_m,f_{\tilde{m}^*}$, it is necessary that $\mathcal{G}(f_m,[\lfloor n/2 \rfloor+1, n])= \mathcal{G}(f_{\tilde{m}},[\lfloor n/2 \rfloor+1, n]) =\mathcal{G}(f_{\tilde{m}^*},[\lfloor n/2 \rfloor+1, n])$. However, by Proposition~\ref{prop:f-dual}, there is a maximal interval between $f_{\tilde{m}},f_{\tilde{m}^*}$ that is contained in $[\lfloor n/2 \rfloor+1, n]$, leading to a contradiction. Therefore, $\mathcal{G}(f_m,[\lfloor n/2 \rfloor])= \mathcal{G}(f_{\tilde{m}},[\lfloor n/2 \rfloor])$ or $\mathcal{G}(f_m,[\lfloor n/2 \rfloor])= \mathcal{G}(f_{\tilde{m}^*},[\lfloor n/2 \rfloor])$ for all $m\in A(\bar{w}/2)\setminus\{\tilde{m},\tilde{m}^*\}$.
		
		Let $a,b\in A(\bar{w}/2)\setminus\{\tilde{m},\tilde{m}^*\}$. Suppose $\mathcal{G}(f_a,[\lfloor n/2 \rfloor])\neq \mathcal{G}(f_{b},[\lfloor n/2 \rfloor])$. Without loss of generality, we may assume further that $\mathcal{G}(f_a,[\lfloor n/2 \rfloor])= \mathcal{G}(f_{\tilde{m}},[\lfloor n/2 \rfloor])$ and $\mathcal{G}(f_b,[\lfloor n/2 \rfloor])= \mathcal{G}(f_{\tilde{m}^*},[\lfloor n/2 \rfloor])$. So there is a maximal interval contained $[\lfloor n/2\rfloor]$ between $f_a,f_{\tilde{m}^*}$. Then by the conditions in Theorem~\ref{UniqueThm} we have $\mathcal{G}(f_a,[\lfloor n/2 \rfloor+1, n])= \mathcal{G}(f_{\tilde{m}^*},[\lfloor n/2 \rfloor+1, n])$. Similarly, we also have $\mathcal{G}(f_b,[\lfloor n/2 \rfloor+1, n])= \mathcal{G}(f_{\tilde{m}},[\lfloor n/2 \rfloor+1, n])$. It follows that $a^*\neq b$ and there are two maximal intervals between $f_a,f_b$, contradicting the conditions in Theorem~\ref{UniqueThm}. Therefore, $\mathcal{G}(f_a,[\lfloor n/2 \rfloor])= \mathcal{G}(f_{b},[\lfloor n/2 \rfloor])$ for all $a,b\in A(\bar{w}/2)\setminus\{\tilde{m},\tilde{m}^*\}$.
		
		Lastly, consider the case where $f_{\tilde{m}} =f_{\tilde{m}^*}$. Let $m\in A(\bar{w}/2)\setminus\{\tilde{m},\tilde{m}^*\}$ be such that $f_m\neq f_{\tilde{m}}$.    
		Since $f_m\neq f_{\tilde{m}}$, by definition of $A(\bar{w}/2)$ and the conditions in Theorem~\ref{UniqueThm}, there exists exactly one maximal interval between $f_m,f_{\tilde{m}}$ that is either contained in $[\lfloor n/2 \rfloor]$ or $[\lfloor n/2\rfloor+1,n]$. Without loss of generality, assume $\mathcal{G}(f_m,[\lfloor n/2 \rfloor])\neq  \mathcal{G}(f_{\tilde{m}},[\lfloor n/2 \rfloor])$ and $\mathcal{G}(f_m,[\lfloor n/2 \rfloor+1, n])= \mathcal{G}(f_{\tilde{m}},[\lfloor n/2 \rfloor+1, n])$. By the 180-degree rotational symmetry of $f_m,f_{\tilde{m}}$ and $f_{m^*},f_{\tilde{m}^*}$, we have $\mathcal{G}(f_{m^*},[\lfloor n/2 \rfloor])=  \mathcal{G}(f_{\tilde{m}^*},[\lfloor n/2 \rfloor])$ and $\mathcal{G}(f_{m^*},[\lfloor n/2 \rfloor+1, n])\neq \mathcal{G}(f_{\tilde{m}^*},[\lfloor n/2 \rfloor+1, n])$. Since $f_{\tilde{m}}=f_{\tilde{m}^*}$, we have $\mathcal{G}(f_m,[\lfloor n/2 \rfloor])\neq  \mathcal{G}(f_{{m}^*},[\lfloor n/2 \rfloor])$ and $\mathcal{G}(f_m,[\lfloor n/2 \rfloor+1, n])\neq \mathcal{G}(f_{{m}^*},[\lfloor n/2 \rfloor+1, n])$. So there exist two maximal intervals between $f_m,f_{m^*}$. The remainder of the proof for this case follows similarly to the case where $f_{\tilde{m}}\neq f_{\tilde{m}^*}$.
	
	\section{Proof of Lemma~\ref{le:half-bar-w-0}}\label{app:half-bar-w-0}
		Consider the case where $f_{m},m\in A(\bar w /2)$ are all the same. 
		By Lemma~\ref{le:half-same} there are no branching points in $\mathcal{G}(f_m,[\lfloor n/2\rfloor ])$ for all $m \in A_f(\bar w /2)$. 
		
		Let $m \in A_f(\bar w /2)$. Note that for any $m' \in  A_{f'}(\bar w /2)$, we have $f'_{m'}(\lfloor n/2 \rfloor)=f_{m}(\lfloor n/2 \rfloor)$. So by Proposition~\ref{prop:noBranch-0} we have $\mathcal{G}(f_m,[0,\lfloor n/2\rfloor])=\mathcal{G}(f'_{m'},[0,\lfloor n/2\rfloor])$ for all $m\in A_f(\bar{w}/2)$ and all $m'\in A_{f'}(\bar{w}/2)$. Recall that $m\in A_f(\bar{w}/2)$ if and only if $m^*\in A_f(\bar{w}/2)$, and $\mathcal{G}(f_{m^*},[\lceil n/2\rceil, n])$ is the same as $\mathcal{G}(f_m,[0,\lfloor n/2\rfloor])$ when rotated $180$ degrees about $(n/2,\bar{w}/2)$. Moreover, the same holds for $m'\in A_{f'}(\bar{w}/2)$. It follows that $\mathcal{G}(f_m,[\lceil n/2\rceil, n])=\mathcal{G}(f'_{m'},[\lceil n/2\rceil, n])$ for all $m\in A_f(\bar{w}/2)$ and all $m'\in A_{f'}(\bar{w}/2)$. Thus, $f_m=f'_{m'}$ for all $m\in A_f(\bar{w}/2)$ and all $m'\in A_{f'}(\bar{w}/2)$. Since $f,f'$ are solutions to $M$, we have $|A_f(\bar{w}/2)|=|A_{f'}(\bar{w}/2)|$. Therefore, if $f_{m},m\in A(\bar w /2)$ are all the same, then ${\psi}_0(f) = {\psi}_0({f'})$.

		Consider the case where $f_{m},m\in A(\bar w /2)$ are not all the same. 
		By Lemma~\ref{le:half-diff-0}, there exists $m_1\in A(\bar w /2)$ such that there is a maximal interval $[l_1+1,l_2-1]\subset [n]$ between $f_{m_1},f_{m_1^*}$, where $l_2\leq \lfloor n/2 \rfloor$. Moreover by Lemma~\ref{le:half-diff}, $(l_2,f_{m_1}(l_2))$ is the only branching point in $\mathcal{G}(f_{m}, [\lfloor n/2\rfloor ])$ and there is no merging point in $\mathcal{G}(f_{m}, [l_2, \lfloor n/2\rfloor ])$ for all $m\in A(\bar{w}/2)$.
		
		Let $m \in A_f(\bar w /2)$. Note that $(l_2,f_{m_1}(l_2))$ is the branching point in $\mathcal{G}(f_m)$ such that $l_2\leq l$ for any branching point $(l,w)\in\mathcal{G}(f_m)$. Similarly to the proof of Lemma~\ref{le:half-bar-w}, let $(l_*,w_*)=(l_2,f_{m_1}(l_2))$ and
		\begin{align*}
			r=\begin{cases}
				0 & \text{if } w_* = {f_{m}(l_*-1)},\\
				1 & \text{if } w_* = {f_{m}(l_*-1)}+1.
			\end{cases}
		\end{align*}
		By definition of $r$, we have $m\in A_f(\bar{w}/2) \cap A_{f}(l_*, w_*) \cap A_{f}(l_*-1, w_*-r)$. 
		
		Let $m'\in A_{f'}(\bar{w}/2) \cap  A_{f'}(l_*, w_*) \cap A_{f'}(l_*-1, w_*-r)$. In the following we will show $\mathcal{G}(f'_{m'},[\lfloor n/2 \rfloor])=\mathcal{G}(f_{m},[\lfloor n/2 \rfloor])$. For notational convenience, let us define $\hat{f}_m=\mathcal{G}(f_{m},[\lfloor n/2 \rfloor])$ and $\hat{f}'_{m'}= \mathcal{G}(f'_{m'},[\lfloor n/2 \rfloor])$. Further, define $\hat{\psi}_0(f)=\{\hat{f}_m \mid m\in A_f(\bar{w}/2) \}$ and $\hat{\psi}_0(f')=\{\hat{f}'_{m'} \mid m'\in A_{f'}(\bar{w}/2) \}$.
		
		By definition of $l_*$, there is no branching point in $\mathcal{G}(f_{m},[ l_*-1])$. Since $f_{m}(l_*-1) = f'_{m'}(l_*-1)$, by Proposition~\ref{prop:noBranch-0} we have $f_{m}(l)=f'_{m'}(l)$ for any $l\in [0, l_*-1]$. 
		Suppose $f_{m}(l)\neq f'_{m'}(l)$ for some $l\in [l_*, \lfloor n/2 \rfloor ]$. Then by Proposition~\ref{prop:noBranch}, there is a merging point in $\mathcal{G}(f_{m},{[l_*, l-1]})$. But there is no merging point in $\mathcal{G}(f_{m}, [l_*=l_2, \lfloor n/2\rfloor ])$, which is a contradiction.			
		Thus, $f_m(l)= f'_{m'}(l)$ for all $l\in [l_*, \lfloor n/2\rfloor]$. It follows that $\hat{f}'_{m'}=\hat{f}_m$ for any $m'\in A_{f'}(\bar{w}/2)\cap A_{f'}(l_*, w_*) \cap A_{f'}(l_*-1, w_*-r)$.
		
		{In the following, we will show $A_{f'}(\bar{w}/2) \supset A_{f'}(l_*, w_*) \cap A_{f'}(l_*-1, w_*-r)$. Toward a contradiction, suppose that there exists $m_0 \in ([2h] \setminus A_{f'}(\bar{w}/2)) \cap A_{f'}(l_*, w_*) \cap A_{f'}(l_*-1, w_*-r)$. Then { $f'_{m_0}(l_*) = f_m{(l_*})$ and $f'_{m_0}(l_*-1) = f_{m}(l_*-1)$}. Note that $\med(f'_{m_0}) \neq \med(f_m)$. {Since $l_*\leq \lfloor n/2\rfloor$}, by Proposition~\ref{prop:noBranch}, there must be a merging point in $\mathcal{G}(f_m, [l_*, \lfloor n/2 \rfloor])$, which contradicts Lemma~\ref{le:half-diff}. {Therefore,} $A_{f'}(\bar{w}/2) \supset A_{f'}(l_*, w_*) \cap A_{f'}(l_*-1, w_*-r)$.}
		
		Note that for any $m'\in  A_{f'}(\bar{w}/2)\setminus( A_{f'}(l_*, w_*) \cap A_{f'}(l_*-1, w_*-r))$, we have $\hat{f}'_{m'}\neq \hat{f}_m$. So the multiplicity of $\hat{f}_{m}$ in $\hat{\psi}_0(f')$ is {$| A_{f'}(l_*, w_*) \cap A_{f'}(l_*-1, w_*-r)|$}. Taking $f'=f$, one can repeat the above arguments to show that $\hat{f}_m=\hat{f}_{\tilde{m}}$ for any $\tilde{m}\in A_{f}(\bar{w/2})\cap A_{f}(l_*, w_*) \cap A_{f}(l_*-1, w_*-r)$ and the multiplicity of $\hat{f}_m$ in $\hat{\psi}_0(f)$ is {$|  A_{f}(l_*, w_*) \cap A_{f}(l_*-1, w_*-r)|$}.
		
		Since $f,f'$ are solutions to $M$, { $| A_{f}(l_*, w_*) \cap A_{f}(l_*-1, w_*-r)| = | A_{f'}(l_*, w_*) \cap A_{f'}(l_*-1, w_*-r)|$}, i.e., the multiplicity of $\hat{f}_m$ in $\hat{\psi}_0(f)$ equals the multiplicity of $\hat{f}_{m}$ in $\hat{\psi}_0(f')$. Furthermore, this holds for distinct $\hat{f}_m\in\hat{\psi}_0(f)$. Since $|A_f(\bar{w}/2)|=|A_{f'}(\bar{w}/2)|$, i.e., $|\hat{\psi}_0(f)|=|\hat{\psi}_0(f')|$, we obtain $\hat{\psi}_0(f) = \hat{\psi}_0(f')$.
		
		Lastly, by Lemma~\ref{le:half-diff}, it holds that $\hat{f}_a=\hat{f}_{m_1}$ for all $a\in A_f(\bar{w}/2)\setminus\{m_1,m_1^*\}$ or $\hat{f}_a=\hat{f}_{m_1^*}$ for all $a\in A_f(\bar{w}/2)\setminus\{m_1,m_1^*\}$. Thus, the multiplicity of $\hat{f}_m$ in $\hat{\psi}_0(f)$ is either $1$ or $|A_f(\bar{w}/2)|-1$. Without loss of generality, assume the multiplicity of $\hat{f}_m$ in $\hat{\psi}_0(f)$ is $1$. Then the multiplicity of $\hat{f}_{m^*}$ in $\hat{\psi}_0(f)$ is $|A_f(\bar{w}/2)|-1$. Moreover, the multiplicity of $\hat{f}'_m$ in $\hat{\psi}_0(f')$ is $1$ and the multiplicity of $\hat{f}'_{(m')^*}$ in $\hat{\psi}_0(f')$ is $|A_{f'}(\bar{w}/2)|-1=|A_{f}(\bar{w}/2)|-1$. Since $f_m$ (resp., $f'_{m'}$) is the same as $f_{m^*}$ (resp., $f'_{(m')^*}$)  when rotated $180$ degrees about $(n/2,\bar{w}/2)$, we have $f_m=f'_{m'}, f_{m^*}=f'_{(m')^*}$ and the multiplicity of $f_m$ (resp., $f'_{m'}$) equals one in $\psi_0(f)$ (resp., $\psi_0(f')$). Now for all $b\in A_f(\bar{w}/2)\setminus \{m,m^*\}$ and all $b'\in A_{f'}(\bar{w}/2)\setminus \{m',(m')^*\}$, we have $f_b =f_{b'}$ since $\hat{f}_b=\hat{f}'_{b'}$. Hence, we conclude $\psi_0(f) = \psi_0(f')$.
\end{document}